\documentclass[orivec]{llncs} 

\usepackage{proof}
\usepackage{latexsym}
\usepackage{amssymb}
\usepackage{amsmath}

\newcommand\nil{{\bf 0}}
\newcommand\pch[1]{{\,{}_{#1}\oplus \,}}

\def\llarrow{-\!\!\!\longrightarrow}
\def\Llarrow{=\!\!\!\Longrightarrow}
\def\pand{\bigwedge}
\def\psum{\bigoplus}

\newcommand\interp[1]{\lbrack\!\lbrack #1 \rbrack\!\rbrack}

\newcommand\one[3]{#1 \stackrel{#2}{\llarrow} #3 }
\newcommand\ar[1]{\stackrel{#1}{\llarrow}}
\newcommand{\nar}[1]{ \not \stackrel{#1}{\longrightarrow} }
\newcommand\dar[1]{\stackrel{#1}{\Llarrow}}
\newcommand\barb[1]{\downarrow_{#1}}

\def\pmay{\sqsubseteq_{pmay}}
\def\pmust{\sqsubseteq_{pmust}}

\def\probpi{\pi_p}

\newcommand\bigstep[3]{#1 \stackrel{#2}{\Llarrow} #3}

\newcommand\sstep[1]{\stackrel{#1}{\llarrow}}
\newcommand\bstep[1]{\stackrel{#1}{\Llarrow}}

\def\Dcal{{\cal D}}

\def\Fcal{{\cal F}}
\def\Lcal{{\cal L}}
\def\Pcal{{\cal P}}
\def\Rcal{{\cal R}}

\def\land{\wedge}
\newcommand\Ref[1]{{\bf ref}(#1)}
\def\val{{\mathbb V}}

\def\simrel{\sqsubseteq}
\def\simord{\triangleleft_S}
\def\simpreo{\sqsubseteq_S}
\def\failsimord{\triangleleft_{FS}}
\def\failsimpreo{\sqsubseteq_{FS}}
\def\simordl{{\overline{\triangleleft}}}
\def\lleq{\sqsubseteq_\Lcal}
\def\fleq{\sqsubseteq_\Fcal}

\newcommand\pdist[1]{\delta[#1]}
\newcommand\lift[1]{\;\overline{#1}\;}
\newcommand\supp[1]{\lceil {#1} \rceil}

\newcommand\ldia[1]{\langle #1 \rangle}

\title{Characterisations of Testing Preorders for a
Finite Probabilistic $\pi$-Calculus}
\author{Yuxin Deng\inst{1} \and Alwen Tiu\inst{2}}
\institute{
Department of Computer Science and Engineering,
Shanghai Jiao Tong University
\and
Research School of Computer Science,
The Australian National University\\
}

\begin{document}

\maketitle

\begin{abstract}
We consider two characterisations of the may and must testing preorders for a
probabilistic extension of the finite $\pi$-calculus: one based on 
notions of probabilistic weak simulations, and the other on a
probabilistic extension of a fragment of Milner-Parrow-Walker modal
logic for the $\pi$-calculus.  We base our notions of simulations on the
similar concepts used in previous work for probabilistic CSP.
However, unlike the case with CSP (or other non-value-passing
calculi), there are several possible definitions of simulation for the
probabilistic $\pi$-calculus, which arise from different ways of scoping
the name quantification.  We show that in order to capture the
testing preorders, one needs to use the ``earliest'' simulation
relation (in analogy to the notion of early (bi)simulation in the
non-probabilistic case). The key ideas in both characterisations are
the notion of a ``characteristic formula'' of a probabilistic process,
and the notion of a ``characteristic test'' for a formula.  As in an
earlier work on testing equivalence for the $\pi$-calculus by Boreale and De
Nicola, we extend the language of the $\pi$-calculus with a mismatch
operator, without which the formulation of a characteristic test will
not be possible.

\end{abstract}

\section{Introduction}

We consider an extension of a finite version (without replication or recursion)
of the $\pi$-calculus~\cite{Milner92IC2} with a probabilistic choice operator, alongside
the non-deterministic choice operator of the $\pi$-calculus. Such an extension has been
shown to be useful in modelling protocols and their 
properties, see, e.g., \cite{Norman09,Chatzikokolakis07}. 
The combination of both probabilistic and non-deterministic choice has long been a subject
of study in process theories, 
see, e.g., \cite{Hansson90,Yi92,Segala94CONCUR,Deng08LMCS}. 
In this paper, we consider a natural
notion of preorders for the probabilistic $\pi$-calculus, based on 
the notion of {\em testing}~\cite{Nicola84,Hennessy88}.
In this testing theory, one defines a notion of  test, what it means
to apply a test to a process, the outcome of a test, and how the outcomes
of tests can be compared. 
In general, the outcome of a test can be any non-empty set, endowed with
a (partial) order; in the case of the original theory, this is simply a two-element lattice,
with the top element representing success and the bottom element representing failure. 
In the probabilistic case, the set of outcomes is the unit interval [0,1], denoting
probabilities of success, with the standard mathematical ordering $\leq$. 
In the presence of non-determinism, it is natural to consider a set of such probabilities
as the result of applying a test to a process. Two standard approaches for comparing
results of a test are the so-called Hoare preorder, written $\simrel_{Ho}$, and
the Smyth preorder, $\simrel_{Sm}$~\cite{Hennessy82}:
\begin{itemize}
\item $O_1 \simrel_{Ho} O_2$ if for every $o_1 \in O_1$ there exists $o_2 \in O_2$ such that $o_1 \leq o_2.$
\item $O_1 \simrel_{Sm} O_2$ if for every $o_2 \in O_2$ there exists $o_1 \in O_1$ such that $o_1 \leq o_2.$
\end{itemize}
Correspondingly, these give rise to two semantic preorders for processes: 
\begin{itemize}
\item {\em may-testing}: 
$P \simrel_{pmay} Q$ iff for every test $T$, $Apply(T,P) \simrel_{Ho} Apply(T,Q)$  
\item {\em must-testing}:
$P \simrel_{pmust} Q$ iff for every test $T$, $Apply(T,P) \simrel_{Sm} Apply(T,Q)$,
\end{itemize}
where $Apply(T,P)$ refers to the result of applying the test $T$ to process $P$.

We derive two characterisations of both may-testing and must-testing: one based 
on a notion of probabilistic weak (failure) simulation~\cite{Segala94CONCUR}, 
and the other based on a modal logic obtained by extending 
Milner-Parrow-Walker (MPW) modal logic
for the (non-probabilistic) $\pi$-calculus~\cite{Milner93TCS}. 

The probabilistic $\pi$-calculus that we consider here is a variant of 
the probabilistic $\pi$-calculus considered in \cite{Chatzikokolakis07}, 
but extended with the mismatch operator. 
As has already been observed in the testing semantics for the non-probabilistic
$\pi$-calculus~\cite{Boreale95IC},
the omission of mismatch would result in a strictly less discriminating test. 
This is essentially due to the possibility of two kinds of output transitions 
in the $\pi$-calculus,
a bound-output action, which outputs a new name, e.g., $\bar x(w).0$,
and a free-output action, e.g., $\bar x y.0.$ Without the mismatch operator, the two processes
are related via may-testing, because the test cannot 
distinguish between output of a fresh
name and output of an arbitrary name (see \cite{Boreale95IC}).

The technical framework used to prove the main results in this paper
is based on previous works
on probabilistic CSP (pCSP)~\cite{Deng07ENTCS,Deng08LMCS}, an extension of Hoare's CSP~\cite{Hoare85}
with a probabilistic choice operator. 
This allows us to adapt some proofs and results from
\cite{Deng07ENTCS,Deng08LMCS} that are not calculus-specific.
The name-passing feature of the $\pi$-calculus, however, gives rise to several
difficulties not found in the non-name-passing calculi such as pCSP,
and it consequently requires new techniques to deal with.
For instance, there is not a canonical notion of (weak) simulation in the $\pi$-calculus,
unlike the case with pCSP.
Different variants arise from different ways of scoping the name quantification 
in the simulation clause dealing with input transitions, e.g., the ``early'' 
vs. the ``late'' variants of (bi)simulation~\cite{Milner92IC2}. In the case of weak simulation,
one also gets a ``delay'' variant of (bi)simulation~~\cite{Ferrari95,Sangiorgi96,vanGlabbeek96}. 
As we show in Section~\ref{sec:sim}, the right notion of simulation is the early variant, as all other
weak simulation relations are strictly more discriminating than the early one. 
Another difficulty is in proving congruence properties, a prerequisite
for the soundness of the (failure) simulation preorders. The possibility of performing a `close'
communication in the $\pi$-calculus requires a combination of closure
under parallel composition and name restriction (see Section \ref{sec:sound}).
We use the so-called ``up-to'' techniques~\cite{Sangiorgi98MSCS} 
for non-probabilistic calculi to prove these congruences.

We show that $\simrel_{pmay}$ coincides with
a simulation preorder $\simpreo$ and a preorder $\simrel_{\Lcal}$
induced by a modal logic $\Lcal$ extending the MPW logic. 
Dually, the must-testing preorder is shown to coincide with
a failure simulation preorder, $\failsimpreo$, and 
a preorder $\simrel_{\Fcal}$ induced by a modal logic $\Fcal$ extending $\Lcal.$
For technical reasons in proving the completeness result of (failure) simulation,
we make use of testing preorders involving vector-based testing ($\pmay^\Omega$ and 
$\pmust^\Omega$ below). 
The precise relations among these preorders are as follows:
$$
\simpreo ~ \subseteq ~ \pmay ~ = ~ \pmay^\Omega ~ \subseteq ~ \lleq ~ \subseteq ~ \simpreo
$$
$$
\failsimpreo ~ \subseteq ~ \pmust ~ = ~ \pmust^\Omega ~ \subseteq ~ \fleq ~ \subseteq ~ \failsimpreo.
$$
The proofs of these inclusions are subjects of Section~\ref{sec:sound}, Section~\ref{sec:modal}
and Section~\ref{sec:comp}. Let us highlight the characterisations of may-testing preorder.
As with the case with pCSP \cite{Deng08LMCS}, the key idea to the proof
of the inclusion $\lleq\; \subseteq\; \simpreo$  is to show that for each process $P$, there
exists a {\em characteristic formula} $\varphi_P$ such that
if $Q\models \varphi_P$ then $P\simpreo Q$.
The inclusion $\pmay^\Omega \; \subseteq\;  \lleq$ is proved by showing that
for each formula $\varphi$, there exists a {\em characteristic test}
$T_\varphi$ such that for all process $P$, $P \models \varphi$ iff
$P$ passes the test $T_\varphi$ with some threshold testing outcome.

\section{Processes and probabilistic distributions}
\label{sec:pi}

We consider an extension of the  (finite) $\pi$-calculus 
with a probabilistic choice operator, $\pch p$, where $p \in (0,1].$
We shall be using the late version of the operational
semantics, formulated in the reactive style (in the sense of \cite{vanGlabbeek95}) 
following previous work \cite{Deng07ENTCS,Deng08LMCS}. 
The use of the late semantics allows for a straightforward definition
of characteristic formulas (see Section~\ref{sec:modal}), which are used in the
completeness proof. 
So our testing equivalence is essentially a ``late'' testing
equivalence. However, as has been shown in \cite{Ingolfsdottir95,Boreale95IC},
late and early testing equivalences coincide for value-passing/name-passing 
calculi.

We assume a countably infinite set of {\em names}, ranged over by $a,b,x,y$ etc. Given a name $a$, its {\em co-name} is $\bar a.$
We use $\mu$ to denote a name or a co-name. 
Process expressions are generated by the following two-sorted grammar:
\[\begin{array}{rcl}
P & ::= & s \mid P {\pch p} P \\
s & ::= & \nil \mid a(x).s \mid \bar a x.s \mid [x=y]s \mid [x\not = y]s \mid s + s \mid  s | s \mid \nu x.s
\end{array}\]
We let $P,Q,...$ range over process terms defined by this grammar, and
$s,t$ range over the subset $S_p$ comprising only the state-based process
terms, i.e. the sub-sort $s$.

The input prefix $a(x)$ and restriction $\nu x$ are name-binding contructs; $x$ in this case
is a bound name. We denote with $fn(P)$ the set of free names in $P$
and $bn(P)$ the set of bound names. The set of names in $P$ (free or bound)
is denoted by $n(P).$
We shall assume that bound names are different from each other and
different from any free names.
Processes are considered equivalent modulo renaming of bound names.
Processes are ranged over by $P$,$Q$,$R$, etc. 
We shall refer to our probablistic extension of the $\pi$-calculus
as $\probpi.$

We shall sometimes use an $n$-ary version of the binary operators.
For example, we use $\bigoplus_{i \in I} p_iP_i$, where $\sum_{i\in I} p_i = 1$,
to denote a process obtained by several applications of the probabilistic choice
operator. Simiarly, $\sum_{i \in I} P_i$ denotes several applications of 
the non-deterministic choice operator $+.$
We shall use the $\tau$-prefix, as in $\tau.P$, as an abbreviation of
$\nu x(x(y).\nil \mid \bar x x. P),$ where $x,y \not \in fn(P).$

In this paper, we take the viewpoint that a probabilistic process
represents an unstable state that may probabilistically evolve into
some stable states. Formally, we describe 
 unstable states as distributions and stable states as state-based
 processes. 
Note that in a state-based process, probablistic choice can only
appear under input/output prefixes. 
The operational semantics of $\probpi$ will be defined only
for state-based processes.

Probabilistic distributions are ranged over by $\Delta.$
A {\em discrete probabilistic distribution} over a set $S$ is a mapping 
$\Delta : S \rightarrow [0,1]$ with $\sum_{s \in S} \Delta(s) = 1.$
The {\em support} of a distribution $\Delta$, denoted
by $\supp \Delta$, is the set $\{s \mid \Delta(s) > 0 \}.$ 
From now on, we shall restrict to only probabilistic distributions
with finite support, and we let $\Dcal(S)$ denote the collection
of such distributions over $S.$ 
If $s$ is a state-based process, then $\pdist s$ denote 
the point distribution that maps $s$ to $1.$ 
For a finite index set $I$, given $p_i$ and distribution $\Delta_i$, for each $i\in I$, 
such that $\sum_{i\in I} p_i = 1$, we define another probability distribution
$\sum_{i\in I} p_i \cdot \Delta_i$ as 
$(\sum_{i\in I} p_i \cdot \Delta_i)(s) = \sum_{i\in I} p_i \cdot \Delta_i(s),$
where $\cdot$ here denotes multiplication. 
We shall sometimes write this distribution as a summation
$p_1 \cdot \Delta_1 + p_2 \cdot \Delta_2 + \ldots + p_n \cdot \Delta_n$
when the index set $I$ is $\{1,\ldots,n\}.$

Probabilistic processes are interpreted
as distributions over state-based processes as follows.
\[\begin{array}{rcl}
\interp s & ::= & \pdist s \ \mbox{ for $s\in S_p$}\\
\interp {P \pch p Q} & ::= & p \cdot \interp P + (1-p) \cdot \interp Q
\end{array}\]
Note that for each process term $P$ the distribution $\interp P$ is
finite, that is it has finite support. 

A transition judgment can take one of the following forms:
$$
\one{s}{a(x)}{\Delta}
\qquad
\one{s}{\tau}{\Delta}
\qquad
\one{s}{\bar a x}{\Delta}
\qquad
\one{s}{\bar a(x)}{\Delta}
$$
The action $a(x)$ is called a {\em bound-input action};
$\tau$ is the silent action; $\bar ax$ is a {\em free-output action}
and $\bar a(x)$ is a {\em bound-output action}.
In actions $a(x)$ and $\bar a(x)$, $x$ is a bound name.
Given an action $\alpha$, we denote with $fn(\alpha)$ the set
of free names in $\alpha$, i.e., those names in $\alpha$ which
are not bound names. The set of bound names in $\alpha$
is denoted by $bn(\alpha)$, and the set of all names (free and bound)
in $\alpha$ is denoted by $n(\alpha).$
The free names of a distribution is the union of free names
of its support, i.e., 
$
fn(\Delta) = \bigcup \{fn(s) \mid s \in \supp \Delta \}.
$

A substitution is a mapping from names to names; substitutions are
ranged over by $\rho, \sigma$ and $\theta.$ 
A substitution $\theta$ is a {\em renaming substitution}
if $\theta$ is an injective map, i.e., $\theta(x) = \theta(y)$
implies $x = y$. 
A substitution is extended to a mapping between processes
in the standard way, avoiding capture of free variables.
We use the notation $s[y/x]$ to denote the result of substituting free
occurrences of $x$ in $s$ with $y.$
Substitution is lifted to a mapping between distributions as follows:
$$
\Delta[y/x] (s) =
\sum \{\Delta(s') \mid s'[y/x] = s \}.
$$
It can be verified that $\interp {P[y/x]} = \interp P [y/x]$
for every process $P.$

The operational semantics is given in Figure~\ref{fig:pi}.
The rules for parallel composition and restriction use an obvious notation for
distributing an operator over distributions, for example:
\[\begin{array}{rcl}
(\Delta_1 ~|~ \Delta_2)(s) & = &
\left\{
\begin{array}{ll}
\Delta_1(s_1) \cdot \Delta_2(s_2) & \hbox{ if $s = s_1 | s_2$ } \\
0 & \hbox{ otherwise}
\end{array}
\right .
\\
(\nu x.\Delta)(s) & = & 
\left\{
\begin{array}{ll}
\Delta(s') & \hbox{ if $s = \nu x.s'$ } \\
0 & \hbox{ otherwise.}
\end{array}
\right.
\end{array}\]
The symmetric counterparts of \textbf{Sum}, \textbf{Par}, \textbf{Com} and \textbf{Close}
are omitted. 
The semantics of $\pi_p$ processes 
is presented 
in terms of simple probabilistic automata \cite{Segala94CONCUR}.

\begin{figure}[t]
$$
\begin{array}{l}
\infer[\textbf{Act}]
{\one {\alpha.P}{\alpha}{\interp{P}}}{}
\hspace{2.2cm}
\infer[\textbf{Sum}]
{\one{s + t}{\alpha}{\Delta}}
{
 \one{s}{\alpha}{\Delta}
}
\medskip\\

\infer[\textbf{Match}]
{\one {[x=x]s} {\alpha} {\Delta}}
{\one s \alpha \Delta}
\qquad\qquad
\infer[\textbf{Mismatch}, x \not = y]
{\one {[x\not = y]s} {\alpha} {\Delta}}
{\one s \alpha \Delta}
\medskip\\

\infer[\textbf{Par}, bn(\alpha)\cap fn(t) = \emptyset]
{\one {s ~ | ~ t}{\alpha} {\Delta ~|~ \pdist t} }
{
 \one {s}{\alpha}{\Delta}
}
\medskip\\

\infer[\textbf{Com}]
{\one {s ~|~ t}{\tau}{\Delta_1[y/x] ~|~ \Delta_2}}
{
 \one {s}{a(x)}{\Delta_1}
 &
 \one {t}{\bar a y}{\Delta_2}
}
\qquad
\infer[\textbf{Close}]
{\one{s ~|~ t}{\tau}{\nu w. (\Delta_1 ~|~ \Delta_2)}}
{
 \one{s}{a(w)}{\Delta_1}
 &
 \one{t}{\bar a(w)}{\Delta_1}
}
\medskip\\

\infer[\textbf{Res}, x \not \in n(\alpha)]
{\one {\nu x.s}{\alpha}{\nu x.\Delta}}
{\one {s}{\alpha}{\Delta}}
\qquad
\infer[\textbf{Open}, y \not = x, y \not \in fn(\nu z.s)]
{\one {\nu z.s}{\bar x(y)}{\Delta[y/z]}}
{\one {s}{\bar x z}{\Delta}}
\end{array}
$$

\caption{The operational semantics of $\probpi$.}
\label{fig:pi}
\end{figure}

\section{Testing probabilistic processes}
\label{sec:test}

As standard in testing theories~\cite{Nicola84,Hennessy88,Boreale95IC}, 
to define a test, we introduce
a distinguished name $\omega$ which can only be used in tests and
is not part of the processes being tested.
A {\em test} is just a probabilistic process with possible free
occurrences of the name $\omega$ as channel name in output prefixes,
i.e., a test is a process which may have subterms of the form
$\bar \omega a.P$. Note that the object of the action prefix
(i.e., the name $a$) is irrelevant for the purpose of testing.
Note also that it makes no differences whether the
name $\omega$ appears in input prefixes instead of output prefixes;
the notion of testing preorder will remain the same.  
Therefore we shall often simply write $\omega.P$ to denote $\bar \omega a.P$,
and $\one {P} {\omega} {\Delta}$ to denote
$\one {P}{\bar \omega a}{\Delta}.$
The definitions of may-testing preorder, $\pmay$, and must-testing preorder, $\pmust$, 
have already been given in the introduction, but we left out the definition of
the $Apply$ function. This will be given below. 

Following \cite{Deng07ENTCS}, to define the $Apply$ function, 
we first define a {\em results-gathering function} $\val: S_p \rightarrow {\cal P}([0,1])$
as follows:
$$
\val(s) =
\left\{
\begin{array}{ll}
\{1\} & \qquad \hbox{if $\one s \omega {}$} \\
\bigcup \{\val(\Delta) \mid \one s \tau \Delta \} & \qquad 
\hbox{if $s \not \stackrel{\omega}{\longrightarrow} {}$ but $\one s {\tau} {}$} \\
\{0\} & \qquad \hbox{ otherwise.}
\end{array}
\right.
$$
Here the notation ${\cal P}([0,1])$ stands for the powerset of $[0,1]$, and
we use $\val(\Delta)$ to denote the set of probabilities
$\{\sum_{s\in \supp \Delta} \Delta(s) \cdot p_s \mid p_s\in\val(s)\}$.
The $Apply$ function is then defined as follows:  
given a test $T$ and a process $P$, 
$$
Apply(T,P) = \val(\interp {\nu \vec x.(T~|~P)})
$$
where $\{\vec x\}$ is the set of free names in $T$ and $P$, excluding $\omega.$
So the process (or rather, the distribution) $\nu \vec x.(T ~|~P)$
can only perform an observable action on $\omega.$

\paragraph{Vector-based testing.}
Following \cite{Deng08LMCS}, we introdude another approach of testing called {\em vector-based testing}, which will play an important role in Section~\ref{sec:comp}.
 
Let $\Omega$ be a set of fresh success actions different from any
normal channel names. An $\Omega$-test is a $\pi_p$-process, but
allowing subterms $\omega.P$ for any $\omega\in\Omega$. Applying such
a test $T$ to a process $P$ yields a non-empty set of test
outcome-tuples $Apply^\Omega(T,P)\subseteq [0,1]^\Omega$. 
For each such tuple, its $\omega$-component gives the probability of successfully performing action $\omega$.

To define a results-gathering function for vector-based testing, we need some auxiliary notations. 
For any action $\alpha$ define $\alpha!:[0,1]^\Omega\rightarrow[0,1]^\Omega$ by
\[\alpha!o(\omega)=\left\{\begin{array}{ll}
1 & \mbox{if $\omega=\alpha$}\\
o(\omega) & \mbox{otherwise}
\end{array}\right.\]
so that if $\alpha$ is a success action in $\Omega$ then $\alpha!$ updates the tuple $1$ at that point, leaving it unchanged otherwise, and when $\alpha\not\in\Omega$ the function $\alpha!$ is the identity. For any set $O\subseteq [0,1]^\Omega$, we write $\alpha!O$  for the set $\{\alpha!o \mid o\in O\}$. 
For any set $X$ define its \emph{convex closure} $\updownarrow X$ by
\[\updownarrow X ~:=~ \{\sum_{i\in I}p_i\cdot o_i \mid o_i\in X \mbox{ for each $i\in I$ and $\sum_{i\in I}p_i=1$}\}.\]
Here, $I$ is assumed to be a finite index set. Finally, zero vector $\vec{0}$ is given by $\vec{0}(\omega)=0$ for all $\omega\in\Omega$. Let $S_p^\Omega$ be the set of state-based $\Omega$-tests.
\begin{definition}
\label{def:vector-based-results}
The vector-based results-gathering function $\val^\Omega:S_p^\Omega\rightarrow \Pcal([0,1]^\Omega)$ is given by
\[\val^\Omega(s)~:=~ \left\{\begin{array}{ll}
\updownarrow\bigcup\{\alpha!(\val^\Omega(\Delta)) \mid s\ar{\alpha}\Delta\} & \mbox{if $s\rightarrow$}\\
\{\vec{0}\} & \mbox{otherwise}
\end{array}\right.\]
The notation $s\rightarrow$ means that $s$ is not a deadlock state, i.e. there is some $\alpha$ and $\Delta$ such that $s\ar{\alpha}\Delta$.
For any process $P$ and $\Omega$-test $T$, we define 
$Apply^\Omega(T,P)$ as  $\val^\Omega(\interp{\nu\vec{x}.(T|P)})$, where
$\{\vec x\} = fn(T,P) - \Omega.$
The vector-based may and must preorders are given by
\[\begin{array}{rcl}
P \pmay^\Omega Q & \mbox{ iff } & \mbox{for all $\Omega$-test $T: Apply^\Omega(T,P) \simrel_{Ho} Apply^\Omega(T,Q)$}\\
P \pmust^\Omega Q & \mbox{ iff } & \mbox{for all $\Omega$-test $T: Apply^\Omega(T,P) \simrel_{Sm} Apply^\Omega(T,Q)$}\\
\end{array}\]
where $\simrel_{Ho}$ and $\simrel_{Sm}$ are the Hoare and Smyth preorders on $\Pcal([0,1]^\Omega)$ 
generated from $\leq$ index-wise on $[0,1]^\Omega$.
\end{definition}
Notice a subtle difference between the definition of $\val^\Omega$ above and the definition
of $\val$ given earlier. In $\val^\Omega$, we use {\em action-based testing}, i.e., 
the actual execution of $\omega$ constitutes a success.  
This is in contrast to the {\em state-based testing} in $\val$, where a success is defined
for a state where a success action $\omega$ is possible, without having to 
actually perform the action $\omega.$ In the case where there is no divergence,
as in our case, these two notions of testing coincide; see \cite{Deng08LMCS} for more
details.

The following theorem can be shown by adapting the proof of Theorem 6.6 in \cite{Deng08LMCS}, which states a general property about probabilistic automata \cite{DGMZ07}.
\begin{theorem}\label{thm:multi-uni}
Let $P$ and $Q$ be any $\pi_p$-processes.
\begin{enumerate}
\item $P\pmay^\Omega Q$ iff $P\pmay Q$
\item $P\pmust^\Omega Q$ iff $P\pmust Q$.
\end{enumerate}
\end{theorem}

\section{Simulation and Failure Simulation}
\label{sec:sim}

To define simulation and failure simulation, we need to generalise the transition relations between
states and distributions to those between distributions and distributions. 
This is defined via a notion of lifting of a relation. 

\begin{definition}[Lifting \cite{Deng09CONCUR}]
\label{def:lifting}
Given a relation $\Rcal \subseteq S_p \times \Dcal(S_p)$,
define a {\em lifted relation} $\overline \Rcal \subseteq \Dcal(S_p) \times \Dcal(S_p)$ as the smallest relation that satisfies
\begin{enumerate}
\item $s \Rcal \Theta$ implies $\pdist{s} \lift\Rcal \Theta$
\item (Linearity) $\Delta_i \lift\Rcal \Theta_i$ for all $i\in I$ implies $(\sum_{i\in I}p_i\cdot\Delta_i) \lift\Rcal (\sum_{i\in I}p_i\cdot\Theta_i)$ for any $p_i\in [0,1]$ with $\sum_{i\in I}p_i = 1$.
\end{enumerate}
\end{definition}

The following is a useful properties of the lifting operation.
\begin{proposition}[\cite{Deng07ENTCS}]
\label{prop:lifting}
Suppose $\Rcal \subseteq S \times \Dcal(S)$ and $\sum_{i \in I} p_i = 1.$ If $(\sum_{i\in I} p_i \cdot \Delta_i) \lift \Rcal \Theta$ then 
$\Theta = \sum_{i \in I} p_i \cdot \Theta_i$ for some set of distributions $\Theta_i$ such that
$\Delta_i \lift \Rcal \Theta_i$ for all $i\in I$.
\end{proposition}

For simplicity of presentation, the lifted version of the
transition relation $\sstep{\alpha}$
will be denoted by the same notation as the unlifted version. 
So we shall write $\one \Delta \alpha \Theta$ when $\Delta$
and $\Theta$ are related by the lifted relation from $\sstep{\alpha}.$
Note that in the lifted transition $\one \Delta \alpha \Theta$, 
{\em all} processes in $\supp{\Delta}$ must be able to simultaneously
make the transition $\alpha$. For example, 
$$
\one{
\frac{1}{2} \cdot \pdist{\bar a x.s} + 
\frac{1}{2} \cdot \pdist{\bar a x.t}
}
{\bar a x}
{\frac{1}{2} \cdot \pdist s + \frac{1}{2} \cdot \pdist t}
$$
but the distribution $\frac{1}{2} \cdot \pdist{\bar a x.s} + \frac{1}{2} \cdot \pdist{\bar b x.t}$ 
will not be able to make that transition. 
We need a few more relations to define (failure) simulation: 
\begin{itemize}
\item We write $\one s {\hat \tau} \Delta$ 
to denote either $\one s \tau \Delta$ or $\Delta = \pdist s.$
Its lifted version will be denoted by the same notation, e.g.,
$\one {\Delta_1}{\hat\tau} {\Delta_2}.$
The reflexive-transitive closure of the latter is denoted by
$\stackrel{\hat \tau}{\Longrightarrow}.$
 
\item$\bigstep {\Delta_1} {\hat \alpha} {\Delta_2}$, for $\alpha \not = \tau$, 
iff $\Delta_1 \bstep{\hat\tau} \Delta' \sstep{\alpha} \Delta'' \bstep{\hat\tau} \Delta_2$
for some $\Delta'$ and $\Delta''.$ 

\item We write $s\barb{a}$ to denote $s\ar{a(x)}$, 
and $s\barb{\bar a}$ to denote either $s\ar{\bar{a}(x)}$ or $s\ar{\bar{a}x}$; 
$s\not \barb{\mu}$ stands for the negation.
We write $s\not\barb{X}$ when $s \not \! \ar{\tau}$ and $\forall\mu\in X: s\not\barb{\mu}$, 
and $\Delta\not\barb{X}$ when $\forall s\in\supp{\Delta}:s\not\barb{X}$.
\end{itemize}

\begin{definition}
\label{def:sim}
A relation $\Rcal \subseteq S_p \times \Dcal(S_p)$ is said to be a {\em failure simulation}
if $s \Rcal \Theta$ implies: 
\begin{enumerate}
\item If $\one s {a(x)} {\Delta}$ and $x \not \in fn(s,\Theta)$, 
then for every name $w$,  there exists $\Theta_1$, $\Theta_2$ and 
$\Theta'$ such that 
$$\Theta \bstep{\hat \tau} \Theta_1 \sstep {a(x)} {\Theta_2}, 
\qquad \Theta_2[w/x] \bstep{\hat \tau} \Theta', 
\qquad \hbox{ and } \qquad (\Delta[w/x]) ~ \overline \Rcal ~ \Theta'. 
$$ 

\item If $\one s \alpha \Delta$ and $\alpha$ is
 not an input action, then
there exists $\Theta'$ such that $\bigstep \Theta {\hat \alpha} {\Theta'}$
and $\Delta ~ \overline \Rcal ~ \Theta'$

\item If $s\not\barb{X}$ then there exists $\Theta'$ such that $\Theta\dar{\hat \tau} \Theta'\not\barb{X}$.
\end{enumerate}
We denote with $\failsimord$ the largest failure simulation relation. Similarly, we define \emph{simulation} and $\simord$ by dropping the third clause above.
The {\em simulation preorder} $\simpreo$ and 
\emph{failure simulation preorder} $\failsimpreo$ 
 on process terms are defined by letting
\[\begin{array}{rll}
P \simpreo Q \mbox{ iff } \mbox{there is a distribution
$\Theta$ with $\bigstep {\interp Q}{\hat \tau} \Theta$
and $\interp P ~ \lift{\simord} ~ \Theta.$}\\
P \failsimpreo Q \mbox{ iff } \mbox{there is a distribution
$\Theta$ with $\bigstep {\interp P}{\hat \tau} \Theta$
and $\interp Q ~ \lift{\failsimord} ~ \Theta.$}
\end{array}\]
\end{definition}

Notice the rather unusual clause for input action, where no silent
action from $\Theta_2$ is permitted after the input transition. This is reminiscent of
the notion of {\em delay (bi)simulation}~\cite{Ferrari95,Sangiorgi96,vanGlabbeek96}. 
If instead of that clause, we simply require $\Theta \bstep {\widehat{a(x)}} {\Theta''}$
and $\Delta[w/x] ~ \overline \Rcal ~ \Theta''[w/x]$ then, 
in the presence of mismatch, simulation is not sound w.r.t.
the may-testing preorder, even in the non-probabilistic case.
Consider, for example, the following
processes: 
$$
P = a(x).\bar a b
\qquad
Q = a(x).[x \not = c] \tau.\bar a b
$$
where we recall that $\tau.R$ abbreviates $\nu z.(z(u) ~|~ \bar z z.R)$ for
some $z \not \in fn(R).$
The process $P$ can make an input transition, and regardless of the
value of the input, it can then output $b$ on channel $a.$
Notice that for $Q$, we have
$$
Q \sstep{a(x)} [x \not = c] \tau. \bar a b \sstep{\tau} \nu z(0 ~|~ \bar a b) = Q'.
$$
$Q'$ can also outputs $b$ on channel $a$, so under this alternative definition, 
$Q$ can simulate $P.$ But $P \not \simrel_{pmay} Q$, as the test
$\bar a c.a(y).\omega$ will distinguish them.
This issue has also appeared in the theory of weak (late) bisimulation
for the non-probabilistic $\pi$-calculus; see, e.g., \cite{Sangiorgi01book}.

Note that the above definition of $\simord$ is what is usually called the
``early'' simulation. 
One can obtain different variants of ``late'' simulation using different
alternations of the universal quantification on names and the existential quantifications on distributions in
clause 1 of Definition~\ref{def:sim}. 
Any of these variants leads to a strictly more discriminating simulation. To see why,
consider the weaker of such late variants, i.e., one in which the universal quantifier
on $w$ comes after the existential quantifier on $\Theta_1$:
\begin{quote}
If $\one s {a(x)} {\Delta}$ and $x \not \in fn(s,\Theta)$, 
then there exists $\Theta_1$ such that for every name $w$,  there exist $\Theta_2$ and 
$\Theta'$ such that 
$$\Theta \bstep{\hat \tau} \Theta_1 \sstep {a(x)} {\Theta_2}, 
\qquad \Theta_2[w/x] \bstep{\hat \tau} \Theta', 
\qquad \hbox{ and } \qquad (\Delta[w/x]) ~ \overline \Rcal ~ \Theta'. 
$$ 
\end{quote}
Let us denote this variant with $\simrel_{S'}.$ Consider the following processes:
$$
P = a(x).\bar b x.\nil + a(x).\nil + a(x).[x=z]\bar b x.\nil
\qquad
Q = \tau.a(x).\bar b x.\nil + \tau.a(x).\nil
$$
It is easy to see that $P \simrel_S Q$ but $P {\not \simrel}_{S'} Q.$

If we drop the silent transitions $\Theta_2[w/x] \bstep{\hat \tau} \Theta'$
in clause (1) of Definition~\ref{def:sim}, i.e., we let $\Theta' = \Theta_2[w/x]$ 
(hence, we get a delay simulation), then again we get a strictly stronger relation than $\simrel_S$. 
Let us refer to this stronger relation as $\simrel_{D}$. 
Let $P$ be $a(x).(c {\pch {\frac{1}{2}}} d)$ and let $Q$ be $a(x).\tau.(c {\pch {\frac{1}{2}}} d).$
Here we remove the parameters in the input prefixes $c$ and $d$ to simplify presentation. 
Again, it can be shown that $P \simrel_S Q$ but $P ~ {\not \simrel}_{D} ~ Q.$ For the latter to hold,
we would have to prove 
$
\frac{1}{2} \cdot \pdist c + \frac{1}{2} \cdot \pdist d ~ \overline{\simord} ~ \pdist{\tau.(c {\pch {\frac{1}{2}}} d)},
$
which is impossible.

Note that (failure) simulation is a relation between processes and distributions,
rather than between processes, so it is not immediately obvious that it is a preorder. 
This is established in Corollary~\ref{cor:sim.fail.preorder} below, whose proof requires a series of lemmas. 

In the following, when we apply a substitution to an action,
we assume that the substitution affects both the free and the bound names
in the action. For example, if $\alpha = a(x)$ and 
$\theta = [b/a, y/x]$ then $\alpha\theta = b(y).$
However, application of a substitution to processes or distributions
must still avoid capture.
\begin{lemma}\label{lm:rename}
Suppose $\sigma$ is a renaming substitution.
\begin{enumerate}
\item If $\one s \alpha \Delta$ then $\one {s\sigma}{\alpha \sigma}{\Delta\sigma}.$

\item If $\bigstep {\Delta} {\hat \alpha} {\Delta'}$
then $\bigstep {\Delta\sigma}{\hat \alpha\sigma} {\Delta'\sigma}.$
\end{enumerate}
\end{lemma}

\begin{lemma}
\label{lm:lifted-trans-rename}
Let $I$ be a finite index set, and let $\sum_{i \in I} p_i = 1.$
Suppose $\one {s_i} {a(x_i)} {\Delta_i}$ for each $i \in I$.
Let $x$ be a fresh name not occuring in any of $s_i$, $a(x_i)$ or $\Delta_i.$
Then
$$
\one {\sum_{i \in I} p_i \cdot \pdist {s_i}} {a(x)}{\sum_{i \in I} p_i \cdot \Delta_i[x/x_i]}.
$$
\end{lemma}
Given the above lemma, given transitions $\one {s_i}{a(x_i)}{\Delta_i}$, we can always assume that, 
all the $x_i$'s are the same fresh name, so that when lifting those transitions to
distributions, we shall omit the explicit renaming of individual $x_i.$ This will simplify the
presentation of the proofs in the following. The same remark applies to bound output transitions.

\begin{lemma}
\label{lm:bigstep-dist}
Suppose $\sum_{i \in I} p_i = 1$ and $\Delta_i \bstep {\hat \alpha} \Phi_i$ for each $i \in I,$ where
$I$ is a finite index set. Then 
$$
\sum_{i \in I} p_i \cdot \Delta_i \bstep {\hat \alpha} \sum_{i \in I} p_i \cdot \Phi_i.
$$
\end{lemma}
\begin{proof}
Same as in the proof of Lemma 6.6. in \cite{Deng07ENTCS}. \qed
\end{proof}

\begin{lemma}
\label{lm:sim-refl}
For every state-based process $s$, we have $s \simord \pdist s$ and $s \failsimord \pdist s.$
\end{lemma}
\begin{proof}
Let $\Rcal \subseteq S_p \times \Dcal(S_p)$ be the relation
defined as follows: $s ~ \Rcal ~ \Theta$ iff $\Theta = \pdist s.$
It is easy to see that $\Rcal$ is a simulation and also a failure simulation. 
\qed 
\end{proof}

\begin{lemma}
\label{lm:sim-like1}
Suppose $\Delta ~ \simordl_S ~ \Phi$ and $\one {\Delta} {\alpha} {\Delta'}$, where $\alpha$
is either $\tau$, a free action or a bound output action. Then 
$\one \Phi {\hat \alpha} {\Phi'}$ for some $\Phi'$ such that $\Delta' ~ \simordl_S ~ \Phi'.$
\end{lemma}
\begin{proof}
Similar to the proof of Lemma 6.7 in \cite{Deng07ENTCS}. \qed
\end{proof}

\begin{lemma}
\label{lm:sim-like2}
Suppose $\Delta ~ \simordl_S ~ \Phi$ and $\one{\Delta}{a(x)}{\Delta'}$. 
Then for all name $w$, there exist $\Psi_1$, $\Psi_2$ and $\Psi$ such that
$$
\Phi \bstep {\hat \tau} \Psi_1 \sstep {a(x)} \Psi_2,
\qquad
\Psi_2[w/x] \bstep {\hat \tau} \Psi, 
\qquad \mbox{ and } 
\qquad
(\Delta'[w/x]) ~ \simordl_S ~ \Psi.
$$
\end{lemma}
\begin{proof}
From $\Delta ~ \simordl_S ~ \Phi$ we have that
\begin{equation}
\label{eq:sim-like2-1}
\Delta = \sum_{i \in I} p_i \cdot \pdist{s_i}, \qquad s_i \simord \Phi_i, \qquad
\Phi = \sum_{i \in I} p_i \cdot \Phi_i. 
\end{equation}
and from $\Delta \sstep {a(x)} \Delta'$ we have:
\begin{equation}
\label{eq:sim-like2-2}
\Delta = \sum_{j \in J} q_j \cdot \pdist{t_j}, \qquad t_j \sstep {a(x)} \Theta_j,
\qquad 
\Delta' = \sum_{j \in J} q_j \cdot \Theta_j.
\end{equation}
We assume w.l.o.g. that all $p_i$ and $q_j$ are non-zero. 
Following \cite{Deng07ENTCS}, we define two index sets: 
$I_j = \{ i \in I \mid s_i = t_j \}$ and $J_i = \{ j \in J \mid t_j = s_i \}.$
Obviously, we have
\begin{equation}
\label{eq:sim-like2-3}
\{(i,j) \mid i \in I, j \in J_i\} = \{(i,j) \mid j \in J, i \in J_i \}, \quad \mbox{and} 
\end{equation}
\begin{equation}
\label{eq:sim-like2-4}
\Delta(s_i) = \sum_{j \in J_i} q_j 
\qquad
\Delta(t_j) = \sum_{i \in I_j} p_i.
\end{equation}
It follows from (\ref{eq:sim-like2-4}) that we can rewrite $\Phi$ as
$$
\Phi = \sum_{i \in I} \sum_{j \in J_i} \frac{p_i \cdot q_j}{\Delta(s_i)} \cdot \Phi_i.
$$
Note that $s_i = t_j$ when $j \in I_i.$ 
Since $s_i \simord \Phi_i$, and $s_i = t_j \sstep {a(x)}{\Theta_j}$, we have,
given any name $w$, some $\Phi_{ij}^1$, $\Phi_{ij}^2$ and $\Phi_{ij}$ such that:
\begin{equation}
\label{eq:sim-like2-5}
\Phi_i \bstep {\hat \tau} \Phi_{ij}^1 \sstep {a(x)} \Phi_{ij}^2, \qquad
\Phi_{ij}^2[w/x] \bstep{\hat \tau}{\Phi_{ij}}, \qquad
 \Theta_j[w/x] ~ \simordl_S ~ \Phi_{ij}.
\end{equation}
Let
$$
\Psi_1 = \sum_{i \in I} \sum_{j \in J_i} \frac{p_i \cdot q_j}{\Delta(s_i)} \cdot \Phi_{ij}^1
\qquad
\Psi_2 = \sum_{i \in I} \sum_{j \in J_i} \frac{p_i \cdot q_j}{\Delta(s_i)} \cdot \Phi_{ij}^2
\qquad
\Psi = \sum_{i \in I} \sum_{j \in J_i} \frac{p_i \cdot q_j}{\Delta(s_i)} \cdot \Phi_{ij}.
$$
Lemma~\ref{lm:bigstep-dist} and (\ref{eq:sim-like2-5}) above give us:
$$
\Phi = \sum_{i \in I} \sum_{j \in J_i} \frac{p_i \cdot q_j}{\Delta(s_i)} \cdot \Phi_{i}
\bstep {\hat \tau}
\Psi_1 
\sstep {a(x)} 
\Psi_2
\qquad
\Psi_2[w/x] 
\bstep {\hat \tau} 
\Psi
$$
It remains to show that $\Delta'[w/x] ~ \simordl_S ~ \Psi.$
\begin{align*}
\Delta'[w/x] 
 & =  \sum_{j \in J} q_j \cdot \Theta_j[w/x] & \\  
 & =  \sum_{j \in J} q_j \cdot \sum_{i \in I_j} \frac{p_i}{\Delta(t_j)} \cdot \Theta_j[w/x] & \mbox{ using (\ref{eq:sim-like2-4})} \\
 & =  \sum_{j \in J} \sum_{i \in I_j} \frac{p_i \cdot q_j}{\Delta(t_j)} \cdot \Theta_j[w/x] \\ 
 & =  \sum_{i \in I} \sum_{j \in J_i} \frac{p_i \cdot q_j}{\Delta(s_i)} \cdot \Theta_j[w/x] & \mbox{ using (\ref{eq:sim-like2-3}) } \\ 
 & \simordl_S ~ \sum_{i \in I} \sum_{j \in J_i} \frac{p_i \cdot q_j}{\Delta(t_j)} \cdot \Phi_{ij} = \Psi & 
  \mbox{ using (\ref{eq:sim-like2-5}) and linearity of $\simordl_S$}
\end{align*}
\qed
\end{proof}

\begin{lemma}
\label{lm:sim-like3}
Suppose $\Delta ~ \simordl_S ~ \Phi$ and $\Delta ~ \bstep {\hat \alpha} \Delta'$,
where $\alpha$ is either $\tau$, a free action or a bound output. 
Then $\Phi \bstep {\hat \alpha} \Phi'$ for some $\Phi'$ such that
$\Delta' ~ \simordl_S ~ \Phi'$.
\end{lemma}
\begin{proof}
Similar to the proof of Lemma 6.8 in \cite{Deng07ENTCS}. 
\qed
\end{proof}

\begin{proposition}\label{prop:sim-refl-trans}
The relation $\simordl_S$ is reflexive and transitive.
\end{proposition}
\begin{proof}
Reflexivity of $\simordl_S$ follows from Lemma~\ref{lm:sim-refl}.
To show transitivity, let us define a relation $\Rcal \subseteq S_p \times \Dcal(S_p)$
as follows: 
$
s ~ \Rcal ~ \Theta 
$
iff there exists $\Delta$ such that
$s ~ \simord ~ \Delta$ and $\Delta ~ \simordl_S ~ \Theta.$
We show that $\Rcal$ is a simulation. 

But first, we claim that 
$\Theta ~ \simordl_S ~ \Delta ~ \simordl_S ~ \Phi$ implies $\Theta ~ \overline \Rcal ~ \Phi.$
This can be proved similarly as in the case of CSP (see the proof of Proposition 6.9
in \cite{Deng07ENTCS}).

Now to show that $\Rcal$ is a simulation, there are two cases to consider.
Suppose $s ~ \Rcal ~ \Phi$, i.e., $s ~ \simord ~ \Delta ~ \simordl_S ~ \Phi.$
\begin{itemize}
\item Suppose $s \sstep {\alpha} \Theta$, where $\alpha$ is either $\tau$,
a free action or a bound output action.
From $s ~ \simord ~ \Delta$, we have 
\begin{equation}
\label{eq:sim-trans-1}
\Delta \bstep {\hat \alpha} \Delta' \qquad \mbox{ and } \qquad
\Theta ~ \simordl_S ~ \Delta'.
\end{equation}
By Lemma~\ref{lm:sim-like3} and (\ref{eq:sim-trans-1}), we have
$\Phi \bstep{\hat \alpha} \Phi'$ and $\Delta' ~ \simordl_S ~ \Phi'$,
and by the above claim and (\ref{eq:sim-trans-1}), $\Theta ~ \overline \Rcal ~ \Phi'$.

\item Suppose $s \sstep {a(x)} \Theta,$ so we have: for all $w$, there exist
$\Delta_1$, $\Delta_2$, and $\Delta'$ such that
\begin{equation}
\Delta \bstep{\hat \tau} \Delta_1 \sstep{a(x)} \Delta_2, \qquad
\Delta_2[w/x] \bstep{\hat \tau} \Delta', \qquad \mbox{ and }
\Theta[w/x] ~ \simordl_S ~ \Delta'.
\end{equation}
Since $\Delta ~ \simordl_S ~ \Phi$, by Lemma~\ref{lm:sim-like3} we have
$\Phi \bstep {\hat \tau} \Phi_1$ and $\Delta_1 ~ \simordl_S ~ \Phi_1.$
And since $\Delta_1 \sstep{a(x)} \Delta_2$, by Lemma~\ref{lm:sim-like2},
for all $w$, there exist $\Phi_2$, $\Phi_3$ and $\Phi_4$ such that:
$$
\Phi_1 \bstep{\hat\tau} \Phi_2 \sstep{a(x)} \Phi_3,
\qquad
\Phi_3[w/x] \bstep{\hat\tau} \Phi_4, 
\qquad
\Delta_2[w/x] ~ \simordl_S ~ \Phi_4.
$$
Lemma~\ref{lm:sim-like3}, together with $\Delta_2[w/x]  ~ \simordl_S ~ \Phi_4$
and $\Delta_2[w/x] \bstep{\hat\tau} \Delta'$, implies that
$\Phi_4 \bstep{\hat\tau} \Phi_5$ and $\Delta' ~ \simordl_S ~ \Phi_5$ for some
$\Phi_5.$
From $\Theta[w/x] ~ \simordl_S ~ \Delta'$ and $\Delta' ~ \simordl_S ~ \Phi_5$,
we have $\Theta[w/x] ~ \overline \Rcal ~ \Phi_5.$
Putting it all together, we have: 
$$
\Phi \bstep{\hat\tau} \Phi_2 \sstep{a(x)} \Phi_3, 
\qquad
\Phi_3[w/x] \bstep{\hat\tau} \Phi_5, 
\qquad 
\Theta[w/x] ~ \overline \Rcal ~ \Phi_5.
$$
\end{itemize}
Thus $\Rcal$ is indeed a simulation. 
\qed
\end{proof}

\begin{proposition}
\label{prop:failsim-refl-trans}
The relation $\simordl_{FS}$ is reflexive and transitive.
\end{proposition}
\begin{proof}
Reflexivity of $\simordl_{FS}$ follows from Lemma~\ref{lm:sim-refl}.
To show transivity, we use a similar argument as in
the proof of Proposition~\ref{prop:sim-refl-trans}:
define $\Rcal$ such that
$s ~ \Rcal ~ \Theta$ iff there exists $\Delta$ such that
$s ~ \failsimord ~ \Delta$ and $\Delta ~ \simordl_{FS} ~ \Theta.$ We show that $\Rcal$ is a failure simulation.

Suppose $s ~ \Rcal ~ \Theta$. 
The matching up of transitions between $s$ and $\Theta$ 
is proved similarly to the case with simulation, by
proving the analog of Lemmas~\ref{lm:sim-like1} - \ref{lm:sim-like3} for failure simulation. 
It then remains to show that
when $s \not \barb{X}$ then there exists
$\Theta'$ such that $\Theta \dar{\hat \tau} \Theta' \not \barb{X}.$
Since $s ~\Rcal ~ \Theta$, by the definition of $\Rcal$,
we have a $\Delta$ s.t. $s ~ \failsimord ~ \Delta$
and $\Delta ~ \simordl_{FS} ~ \Theta.$ The former
implies that $\Delta \dar{\hat \tau} \Delta' \not \barb{X}$,
for some $\Delta'$. 
It can be shown that, using arguments similar to
the proof of Lemma~\ref{lm:sim-like3} that
$\Theta \dar{\hat\tau} \Theta'$ for some $\Theta'$
such that $\Delta' ~ \simordl_{FS} \Theta'.$
Suppose $\supp {\Delta'} = \{s_i\}_{i \in I},$ i.e.,
$\Delta' = \sum_{i\in I} p_i \cdot \pdist{s_i}$
with $\sum_{i\in I} p_i = 1.$
Obviously, $s_i \not \barb{X}$ for each $i\in I.$
By Proposition~\ref{prop:lifting}, 
$\Theta = \sum_{i\in I} p_i \cdot \Theta_i$
for some distributions $\Theta_i$ such that
$\pdist{s_i} ~ \simordl_{FS} ~ \Theta_i.$  
The latter implies, by Definition~\ref{def:lifting},
that $s_i ~ \failsimord ~ \Theta_i.$ 
Since $s_i \not \barb{X}$, it follows that
$\Theta_i \dar{\hat \tau} \Theta_i' \not \barb{X}$,
for some $\Theta_i'.$ 
Thus $\Theta \dar{\hat \tau} (\sum_{i \in I} p_i \cdot \Theta_i) \not \barb{X}.$
\qed
\end{proof}

\begin{corollary}\label{cor:sim.fail.preorder}
The relations $\simrel_S$ and $\failsimpreo$ are preorders.
\end{corollary}
\begin{proof}
The fact that $\simrel_S$ is a preorder
follows from Lemma~\ref{lm:sim-like3} and Proposition~\ref{prop:sim-refl-trans}. 
Similar arguments hold for $\failsimpreo$, using an analog of
Lemma~\ref{lm:sim-like3} and Proposition~\ref{prop:failsim-refl-trans}.
\qed
\end{proof}

\section{Soundness of the simulation preorders}
\label{sec:sound}

In proving soundness of the simulation preorders with respect to
testing preorders, we first need to prove
certain congruence properties, i.e., closure
under restriction and parallel composition. For this, it is helpful
to consider a slightly more general definition of simulation,
which incorporates another relation. This technique, called
the {\em up-to} technique, has been used in the literature
to prove congruence properties of various (pre-)order for
the $\pi$-calculus~\cite{Sangiorgi98MSCS}.

\begin{definition}[Up-to rules]
Let $\Rcal \subseteq S_p \times \Dcal(S_p).$
Define the relation $\Rcal^t$ where $t \in \{r, \nu, p\}$
as the smallest relation which satisfies
the closure rule for $t$, given below (where $\sigma$ is a renaming substitution):
$$
\infer[r]
{s \sigma ~\Rcal^r ~ \Delta \sigma}
{s ~ \Rcal ~ \Delta}
\qquad
\infer[\nu]
{(\nu \vec x.s) ~ \Rcal^\nu ~ (\nu \vec x.\Delta)}
{s ~ \Rcal ~ \Delta}
\qquad
\infer[p]
{(s_1 ~|~ s_2) ~ \Rcal^p ~ (\Delta_1 ~|~ \Delta_2)}
{
 s_1 ~ \Rcal ~ \Delta_1 
& 
 s_2 ~ \Rcal ~ \Delta_2
}
$$
\end{definition}

\begin{definition}[(Failure) Simulation up-to]
A relation $\Rcal \subseteq S_p \times \Dcal(S_p)$ is said to be a {\em (failure) simulation up to 
renaming} (likewise, restriction and parallel composition) if 
it satisfies the clauses 1, and 2, (and 3 for failure simulation) 
in Definition~\ref{def:sim}, but with
$\overline \Rcal$ in the clauses  replaced by 
$\overline {\Rcal^r}$ (respectively, $\overline{\Rcal^\nu}$ and
$\overline{\Rcal^p}$).
\end{definition}

It is easy to see that $\Rcal \subseteq \Rcal^t$ for any $t \in \{r,\nu\}$
(i.e., via the identity relation as renaming substitution in the former, and via the empty restriction
in the latter). The following lemma is then an easy consequence.

\begin{lemma}
\label{lm:sim-upto}
If $\Rcal$ is a (failure) simulation then it is a (failure) simulation up-to renaming,
and also a (failure) simulation up to restriction.
\end{lemma}

Our objective is really to show that simulation up-to parallel composition
is itself a simulation. This would then entail that (the lifted) simulation
is closed under parallel composition, from which soundness w.r.t. may-testing
follows. We prove this indirectly in three stages: 
\begin{itemize}
\item simulation up-to renaming is a simulation;
\item simulation up-to restriction is a simulation up-to renaming (hence also
a simulation by the previous item); 
\item and, finally, simulation up-to parallel composition is a simulation up-to
restriction.
\end{itemize}

\subsection{Up to renaming}

Note that as a consequence of Lemma~\ref{lm:rename} (1),  
given an injective renaming substitution
$\sigma$, we have: 
if $\one {s\sigma}{\alpha'}{\Delta'}$
then there exists $\alpha$ and $\Delta$ such that $\alpha' = \alpha \sigma$,
$\Delta' = \Delta\sigma$ and $\one s \alpha \Delta.$ This is proved by
simply applying Lemma~\ref{lm:rename} (1) to $\one {s\sigma}{\alpha'}{\Delta'}$
using the inverse of $\sigma$.

In the following, we shall write $\Rcal^{tt}$ to denote 
$(\Rcal^t)^t$, i.e., the result of applying the up-to closure
rule $t$ twice to $\Rcal.$

\begin{lemma}
$\Rcal^{rr} = \Rcal^r.$
\end{lemma}

\begin{lemma}
\label{lm:lift-renaming}
If $\Delta_1 ~ \overline{\Rcal^r} ~ \Delta_2$ then
$(\Delta_1\sigma) ~ \overline {\Rcal^r} ~ (\Delta_2 \sigma)$ 
for any renaming substitution $\sigma.$
\end{lemma}
\begin{proof}
This follows from the fact that $\Delta_1 ~ \overline{\Rcal^r} ~ \Delta_2$
implies $\Delta_1\sigma ~ \overline{\Rcal^{rr}} ~ \Delta_2\sigma$
and that $\Rcal^{rr} = \Rcal^r.$
\qed
\end{proof}

\begin{lemma}
\label{lm:upto-renaming}
If $\Rcal$ is a (failure) simulation up to renaming, then 
$\Rcal^r \subseteq \simord$ (respectively, $\Rcal^r \subseteq \failsimord$).
\end{lemma}
\begin{proof}
Suppose $\Rcal$ is a simulation. It is enough to show that $\Rcal^r$ is a simulation. 
So suppose $s ~\Rcal^r ~ \Delta$ and  $\one {s} \alpha \Theta.$
By the definition of $\Rcal^r$, $s = s'\sigma$ and $\Delta = \Delta'\sigma$
for some renaming substitution $\sigma$ and some $s'$ 
and $\Delta'$ such that $s' ~ \Rcal ~ \Delta'.$
There are several cases to consider depending on the type of $\alpha$.

\begin{itemize}
\item $\alpha$ is $\tau$ or a free action: 
By Lemma~\ref{lm:rename} (1) we have $\one {s'} {\alpha'} {\Theta'}$
for some $\alpha'$ and $\Theta'$ such that $\alpha = \alpha'\sigma$
and $\Theta = \Theta'\sigma.$ Since $\Rcal$ is a simulation up to renaming,
$s' \Rcal \Delta'$ implies that 
$\bigstep {\Delta'}{\hat{\alpha'}}{\Delta_1}$
and $\Theta' ~ \overline{\Rcal^r} ~ \Delta_1.$
The former implies, by Lemma~\ref{lm:rename} (2), that
$\bigstep{\Delta}{\hat \alpha}{\Delta_2}$
for some $\Delta_2$ such that $\Delta_2 = \Delta_1\sigma,$
while the latter implies, by Lemma~\ref{lm:lift-renaming},
that $\Theta = (\Theta'\sigma) ~ \overline {\Rcal^r} ~ (\Delta_1\sigma) = \Delta_2.$

\item $\alpha = a(x)$ for some $a$ and $x$:
In this case, $x \not \in fn(s,\Delta),$ so we can assume, without loss of generality, 
that $x$ does not occur in $\sigma.$ 
Using a similar argument as in the previous case, we have that
$\one {s'} {b(x)} {\Theta'}$
for some $b$ and $\Theta'$ such that $\sigma(b) = a$
and $\Theta = \Theta'\sigma.$
Since $\Rcal$ is a simulation up to renaming,
$s' \Rcal \Delta'$ implies that for every name $w$, there exist
$\Delta_w^1$, $\Delta_w^2$ and $\Delta_w$ such that: 
\begin{equation}
\label{eq:ren1}
\Delta' \bstep {\hat \tau} \Delta_w^1 \sstep {b(x)} \Delta_w^2, 
\qquad
\Delta_w^2[w/x] \bstep {\hat \tau} \Delta_w, \quad \mbox{ and }
\end{equation}
\begin{equation}
\label{eq:ren2}
\Theta'[w/x] ~ \overline{\Rcal^r} ~ \Delta_w.
\end{equation}
Let $\Phi_1 = \Delta_w^1\sigma$, $\Phi_2 = \Delta_w^2 \sigma$
and $\Phi = \Delta_w\sigma.$
From (\ref{eq:ren1}) and Lemma~\ref{lm:rename} (2) we get:
$$
\Delta = \Delta' \sigma \bstep {\hat \tau} 
\Delta_w^1\sigma = \Phi_1 \sstep{a(x)} \Delta_w^2 \sigma = \Phi_2.
$$
By (\ref{eq:ren1}), the freshness assumption of $x$ w.r.t. $\sigma$, 
and Lemma~\ref{lm:rename} (2), we get
$$
\Phi_2[w/x] = \Delta_w^2\sigma [w/x] = \Delta_w^2 [w/x] \sigma 
\bstep{\hat\tau} \Delta_w\sigma = \Phi.
$$
Finally, by (\ref{eq:ren2}) and Lemma~\ref{lm:lift-renaming}, 
$
\Theta[w/x] = \Theta'\sigma [w/x] = 
\Theta'[w/x] \sigma ~ \overline{\Rcal^r} ~ \Delta_w\sigma = \Phi.
$

\item $\alpha = \bar a(x)$: This case can be proved similarly to
the previous cases. 
\end{itemize}
For the case where $\Rcal$ is a failure simulation, we additionally
need to show that whenever $s ~ \Rcal^r ~ \Delta$ and $s \not \barb{X}$,
we have $\Delta \dar{\hat \tau} \Theta \not \barb{X}$ for some $\Theta$.
Since $s \Rcal \Delta$, we have $s = s'\sigma$ and $\Delta = \Delta'\sigma$
for some $s'$, $\Delta$ and renaming substitution $\sigma.$
Let $X' = X\sigma^{-1}$, i.e., $X'$ is the inverse image of $X$ under $\sigma.$
Then we have that $s' \not \barb{X'}$, and $\Delta' \dar{\hat \tau} \Theta' \not \barb{X'}.$
Applying $\sigma^{-1}$ to the latter, 
we obtain 
$\Delta \dar{\hat \tau} \Theta \not \barb{X}.$
\qed
\end{proof}

\begin{lemma}
\label{lm:clo-renaming}
Suppose $P \simrel_S Q$ ($P \simrel_{FS} Q$) 
and $\sigma$ is a 
renaming substitution. Then $P \sigma \simrel_S Q\sigma$
(respectively, $P \sigma \simrel_{FS} Q\sigma$).
\end{lemma}
\begin{proof}
Immediate from Lemma~\ref{lm:upto-renaming}. \qed
\end{proof}

\subsection{Up to name restriction}
The following lemma says that transitions are closed under name
restriction, if certain conditions are satisfied.
\begin{lemma}\label{lm:res}
\begin{enumerate}
\item For every state-based process $s$, every action $\alpha$
and every list of names $\vec x$ such that 
$\{\vec x\}\cap n(\alpha) = \emptyset$,
$\one{s}{\alpha}{\Delta}$
implies $\one{\nu \vec x.s}{\alpha}{\nu \vec x.\Delta}.$

\item For every $\Delta$ and $\Phi$, every action $\alpha$
and every list of names $\vec x$ such that 
$\{\vec x\}\cap n(\alpha) = \emptyset$,
$\Delta \sstep{\alpha} \Phi$
implies $\nu \vec x.\Delta \sstep{\alpha} \nu \vec x.\Phi.$

\item Suppose $\one s {\bar a b} \Delta$ and suppose $\vec x$ and $\vec y$
are names such that $\{\vec x,\vec y\} \cap \{a,b\} = \emptyset.$
Then $\one {\nu \vec x \nu b\nu \vec y. s} {\bar a(b)}{\nu \vec x\nu \vec y.\Delta}$.
\end{enumerate}
\end{lemma}

\begin{lemma}
\label{lm:nu}
If $\Delta ~ \overline{\Rcal^\nu} ~ \Theta$ then 
$(\nu \vec x.\Delta) ~ \overline{\Rcal^\nu} ~ (\nu \vec x.\Theta) $
\end{lemma}

\begin{lemma}
\label{lm:upto-restriction}
If $\Rcal$ is a (failure) simulation up to restriction, then 
$\Rcal^\nu \subseteq \simord$ (respectively, $\Rcal^\nu \subseteq \failsimord$).
\end{lemma}
\begin{proof}
Suppose $\Rcal$ is a simulation up to restriction. 
We show that $\Rcal^\nu$ is a simulation up to renaming,
hence by Lemma~\ref{lm:upto-renaming} we have 
$\Rcal^\nu \subseteq \Rcal^{\nu r} \subseteq \simord.$

Suppose $s ~ \Rcal^\nu \Delta$
and $\one {s} {\alpha}{\Theta}.$
By the definition of $\Rcal^\nu$, we have that
$s = \nu \vec x.s'$, $\Delta = \nu \vec x.\Delta'$,
and $s'[\vec y/\vec x] ~ \Rcal ~ \Delta'[\vec y/\vec x]$
for some $\vec y$ such that
$\{\vec y\} \cap fn(s,\Delta) = \emptyset.$

There are several cases depending on how the transition 
$\one s \alpha \Theta$ is derived. 
Note that there may be implicit $\alpha$-renaming involved
in the derivations of a transition judgment.
We assume that the names $\vec x$ are chosen such that
no $\alpha$-renaming is needed in deriving the transition
relation $\one{\nu \vec x.s'}{\alpha}{\Theta}$,
e.g., one such choice would be one that avoids clashes with the free 
names in $\vec y$, $s$, and $\Delta$. 

\begin{itemize}
\item $\alpha$ is either $\tau$
or a free action. In this case, the transition must have been
derived as follows:
$$
\infer=[res]
{\one {\nu \vec x.s'}{\alpha}{\nu \vec x.\Theta'}}
{
 \one {s'}{\alpha}{\Theta'}
}
$$
where $\Theta = \nu \vec x.\Theta'$
and $n(\alpha) \cap \{\vec x\} = \emptyset.$
Here a double-line in the inference rule indicates zero or more applications
of the rule. 
An inspection on the operational semantics will reveal that
in this case, $n(\alpha) \subseteq fn(s)$ and $fn(\Theta) \subseteq fn(s)$. 
So in particular, $\{ \vec y\} \cap n(\alpha) = \emptyset.$
We thus can apply the renaming substitution $[\vec y/\vec x, \vec x/\vec y]$
to get
$
\one {s'[\vec y/\vec x]} {\alpha} {\Theta'[\vec y/\vec x]}.
$
Since $s'[\vec y/\vec x] ~ \Rcal ~ \Delta'[\vec y/\vec x]$, we have
that
$
\bigstep{\Delta'[\vec y/\vec x]}{\alpha}{\Delta''[\vec y/\vec x]}
$
and
$
\Theta'[\vec y/\vec x] ~ \overline{\Rcal^\nu} ~ \Delta''[\vec y/\vec x].
$
The former implies, via Lemma~\ref{lm:res} (1), that
$
\bigstep{\nu \vec x. \Delta'}{\alpha}{\nu \vec x.\Delta''}
$ 
and the latter implies, via Lemma~\ref{lm:nu},
that 
$
(\nu \vec x. \Theta') ~ \overline{\Rcal^\nu} ~ (\nu \vec x. \Delta'')
$. 
Since $\Rcal^\nu \subseteq (\Rcal^{\nu})^r$, we also have
$
(\nu \vec x. \Theta') ~ \overline{\Rcal^{\nu r}} ~ (\nu \vec x. \Delta'').
$

\item $\alpha = a(z)$: 
With a similar argument as in the previous case, we can show that
in this case we must have $\one {s} {a(z)}{\Theta'}$
where $\Theta = \nu \vec x.\Theta'.$ We need to show that
for every name $w$, there exist $\Gamma_w^1$, $\Gamma_w^2$ and $\Gamma_w$ such that
$\Delta \bstep {\hat \tau} \Gamma_w^1 \sstep {a(z)} \Gamma_w^2$, 
$\Gamma_w^2[w/z] \bstep{\hat \tau} \Gamma_w$,
and $\Theta[w/z] ~ \overline{\Rcal^{\nu r}} ~ \Gamma_w.$

Note that $z \not \in \{\vec x\}$,
but it may be the case that $z \in \{\vec y\}.$
So we first apply a renaming $[u/z,z/u, \vec y/\vec x,\vec x/\vec y]$, for some fresh name $u$,  
to the transition $\one {s'}{a(z)}{\Theta'}$ to get:
$$
\one{s'[\vec y/\vec x]}{a(u)}{\Theta'[u/z,\vec y/\vec x]}.
$$

Since $s'[\vec y/\vec x] ~ \Rcal ~ \Delta'[\vec y/\vec x]$, 
we have, for every name $w$, some $\Delta_w^1$, $\Delta_w^2$ and $\Delta_w$
such that 
\begin{equation}
\label{eq:clo-nu2a}
\Delta'[\vec y/\vec x] \bstep{\hat\tau} 
\Delta_w^1 \sstep{a(u)} \Delta_w^2, 
\qquad
\Delta_w^2[w/u] \bstep{\hat \tau} \Delta_w, \qquad \mbox{and }
\end{equation}
\begin{equation}
\label{eq:clo-nu3a}
\Theta'[u/z, \vec y/\vec x][w/u] = \Theta'[w/z,\vec y/\vec x] ~ \overline{\Rcal^\nu} ~ \Delta_w[w/u].
\end{equation}
Let $\Phi_w^1$, $\Phi_w^2$ and $\Phi_w$ be distributions 
such that $\Delta_w^1 = \Phi_w^1[\vec y/\vec x]$,
$\Delta_w^2 = \Phi_w^2[u/z, \vec y/\vec x]$,
and $\Delta_w  = \Phi_w[\vec y/\vec x].$
So in particular, $\Delta_w^2[w/u] = \Phi_w^2[w/z, \vec y/\vec x]$
and $\Delta_w[w/u] = \Phi_w[w/z, \vec y/\vec x].$
Then (\ref{eq:clo-nu2a}) can be rewritten as:
\begin{equation}
\label{eq:clo-nu2b}
\Delta'[\vec y/\vec x] \bstep{\hat\tau} 
\Phi_w^1[\vec y/\vec x]  \sstep{a(u)} \Phi_w^2[u/z, \vec y/\vec x] 
\qquad
\Phi_w^2[w/z, \vec y/\vec x] \bstep{\hat \tau} \Phi_w[\vec y/\vec x],
\end{equation}
and (\ref{eq:clo-nu3a}) can be rewritten as:
\begin{equation}
\label{eq:clo-nu3b}
\Theta'[w/z,\vec y/\vec x] ~ \overline{\Rcal^\nu} ~ \Phi_w[w/z, \vec y/\vec x].
\end{equation}

Now, to define $\Gamma_w^1$, $\Gamma_w^2$ and $\Gamma_w$, we need to consider
two cases, based on the value of $w$. The reason is that in the construction
of $\Gamma_w$ we need to bound the free names in $\Phi_w$, so if $z$ is substituted
with a name in $\vec y$, it could get captured. 
\begin{itemize}
\item $w \not \in \{\vec x, \vec y\}$. 
In this case, define: 
$$
\Gamma_w^1 = \nu \vec x. \Phi_w^1, \qquad
\Gamma_w^2 = \nu \vec x. \Phi_w^2, \qquad
\Gamma_w = \nu \vec x. \Phi_w.
$$
By Lemma~\ref{lm:res} (1) and (\ref{eq:clo-nu2b}), we have:
$$
\nu \vec x. \Delta' \bstep {\hat \tau} \Gamma_w^1 
\sstep {a(z)} \Gamma_w^2, 
\qquad
\Gamma_w^2[w/z] \bstep {\hat \tau} \Gamma_w
$$
and by Lemma~\ref{lm:nu} and (\ref{eq:clo-nu3b}), we have 
$$
(\Theta[w/z]) = (\nu \vec x. \Theta')[w/z] ~ \overline{\Rcal^\nu} ~ 
\Gamma_w,
$$
hence also,
$
(\Theta[w/z]) = (\nu \vec x. \Theta')[w/z] ~ \overline{\Rcal^{\nu r}} ~ 
\Gamma_w.
$

\item $w \in \{\vec x, \vec y\}.$ Let $v$ be a new name (distinct from
all other names considered so far). From the previous case, we
know how to construct $\Gamma_v^1$, $\Gamma_v^2$ and $\Gamma_v$ such that
\begin{equation}
\label{eq:clo-nu3c}
\nu \vec x. \Delta' \bstep {\hat \tau} \Gamma_v^1 
\sstep {a(z)} \Gamma_v^2, 
\qquad
\Gamma_v^2[v/z] \bstep {\hat \tau} \Gamma_v
\qquad
(\Theta[v/z])  ~ \overline{\Rcal^{\nu r}} ~ 
\Gamma_v.
\end{equation}
In this case, let $\Gamma_w^1 = \Gamma_v^1$, $\Gamma_w^2 = \Gamma_v^2$
and $\Gamma_w = \Gamma_v[w/v].$ (Note that because subsitution is capture-avoiding,
the bound names in $\Gamma_v$ will be renamed via $\alpha$-conversion). 
Then by Lemma~\ref{lm:rename} (2) and Lemma~\ref{lm:lift-renaming} 
and (\ref{eq:clo-nu3c}):
$$
\nu \vec x. \Delta' \bstep {\hat \tau} \Gamma_w^1 
\sstep {a(z)} \Gamma_w^2, 
\qquad
\Gamma_v^2[w/z] \bstep {\hat \tau} \Gamma_w
\qquad
(\Theta[w/z])  ~ \overline{\Rcal^{\nu r}} ~ 
\Gamma_w.
$$
\end{itemize}

\item If $\alpha$ is a bound output action, i.e., $\alpha = \bar a(b)$ for some $a$ and $b.$
There are two subcases to consider, depending on whether $b \in \{\vec x\}$ (i.e., one of
the restriction names $\vec x$ is extruded) or not. The latter can be proved similarly to
the previous case. We show here a proof of the former case.
So suppose $b \in \vec x$, i.e., $\nu \vec x = \nu \vec x_1 \nu b \nu \vec x_2$ and
suppose that $[\vec y/\vec x]$ maps $b$ to $c$, i.e., $\nu \vec y = \nu \vec y_1 \nu c \nu \vec y_2.$ Suppose 
the transition relation is derived as follows:
$$
\infer=[res]
{\one{\nu\vec x_1\nu b \nu \vec x_2.s'}{\bar a(b)}{\nu \vec x_1 \nu \vec x_2. \Theta'}}
{
 \infer[open]
 {\one {\nu b \nu \vec x_2.s} {\bar a(b)}{\nu \vec x_2.\Theta'}}
 {
  \infer=[res]
  {\one {\nu \vec x_2.s}{\bar a b}{\nu \vec x_2.\Theta'} }
  {\one {s}{\bar a b}{\Theta'}}
 }
}
$$
Applying the renaming $[\vec y/\vec x, \vec x/\vec y]$ we have:
$
\one{s[\vec y/\vec x]}{\bar a c}{\Theta'[\vec y/\vec x]}.
$
Since $s'[\vec y/\vec x] ~ \Rcal ~ \Delta'[\vec y/\vec x]$, we have that
\begin{equation}
\label{eq:clo-nu4}
\bigstep {\Delta'[\vec y/\vec x]}{\bar a c}{\Phi}, \qquad \mbox{ and } \qquad
\Theta'[\vec y/\vec x] ~ \overline{\Rcal^\nu} ~ \Phi.
\end{equation}
Let $\Psi[\vec y/\vec x] = \Phi.$
Lemma~\ref{lm:res} (3) and (\ref{eq:clo-nu4}) imply that
$$
\bigstep{\nu \vec x.\Delta' = \nu \vec y_1\nu c\vec y_2.\Delta'[\vec y/\vec x]}
{\bar a(c)}{\nu\vec y_1\nu \vec y_2.\Psi[\vec y/\vec x] = \nu \vec x_1\vec x_2. \Psi[c/b]}
$$
and by an application of a renaming (Lemma~\ref{lm:rename} (1))
we get
$$
\bigstep {\nu \vec x.\Delta'}{\bar a (b)}{\nu\vec x_1\nu \vec x_2.\Psi}.
$$
Lemma~\ref{lm:nu} and (\ref{eq:clo-nu4}) imply
$$
(\nu \vec x_1\nu\vec x_2.\Theta'[c/b]) ~ \overline{\Rcal^\nu} ~
(\nu\vec x_1\nu \vec x_2.\Psi[c/b])
$$
hence, via the renaming $[c/b,b/c]$, 
$
(\nu \vec x_1\nu\vec x_2.\Theta') ~ \overline{\Rcal^{\nu r}} ~
(\nu\vec x_1\nu \vec x_2.\Psi).
$
\end{itemize}
If $\Rcal$ is a failure simulation up to restriction, we need to additionally
show that $\Rcal^\nu$ satisfies clause 3 of Definition~\ref{def:sim}.
Suppose $s ~ \Rcal^\nu ~ \Theta$. Then $s = \nu \vec x. s'$ and $\Theta = \nu \vec x.\Theta'$
for some $\vec x$, $s'$ and $\Theta'$ such that $s' ~ \Rcal ~ \Theta'.$
Suppose $s \not \barb{X}.$ We need to show that $\Theta \dar{\hat \tau} \Delta$
such that $\Delta \not \barb{X}$ for some $\Delta.$
Since name restriction hides visible actions, it can be shown that 
$s' \not \barb{X \setminus \{\vec x\}}$ 
iff $\nu \vec x. s' \not \barb{X}.$ 
So from $s' ~ \Rcal ~ \Theta'$ we have that $\Theta' \dar{\hat \tau} \Delta' \not \barb{X \setminus \{\vec x\}}.$
Let $\Delta = \nu \vec x.\Delta'.$
Then by Lemma~\ref{lm:res} (2), we have
$\Theta = \nu \vec x.\Theta' \dar{\hat \tau} \nu \vec x.\Delta' = \Delta \not \barb{X}.$
\qed
\end{proof}

\begin{lemma}
\label{lm:clo-res}
If $P \simrel_S Q$ ($P \simrel_{FS} Q$) then 
$(\nu \vec x. P) ~ \simrel_S (\nu \vec x.Q)$ (respectively, $(\nu \vec x.P) \simrel_{FS} (\nu \vec x.Q)$).
\end{lemma}
\begin{proof}
This is a simple corollary of Lemma~\ref{lm:sim-upto} and Lemma~\ref{lm:upto-restriction}. \qed
\end{proof}

\subsection{Up to parallel composition}

The following lemma will be useful in proving the closure of simulation
under parallel composition. It is independent of the underlying calculus,
and is originally proved in \cite{Deng07ENTCS}.
\begin{lemma}
\label{lm:par0}
\begin{enumerate}
\item $
(\sum_{j\in J} p_j \cdot \Phi_j) ~|~ (\sum_{k\in K} q_k \cdot \Delta_k)
= \sum_{j\in J} \sum_{k \in K} (p_j \cdot q_k) \cdot (\Phi_j ~|~ \Delta_k).
$

\item Suppose $\Rcal, \Rcal' \subseteq S_p \times \Dcal(S_p)$ are two relations
such that $s \Rcal' \Delta$ whenever $s = s_1 ~|~ s_2$ and 
$\Delta = \Delta_1 ~|~ \Delta_2$ with $s_1 \Rcal \Delta_1$
and $s_2 \Rcal \Delta_2.$
Then $\Phi_1 \overline \Rcal \Delta_1$
and $\Phi_2 \overline \Rcal \Delta_2$ imply
$(\Phi_1 ~|~ \Phi_2) \overline {\Rcal'} (\Delta_1 ~|~ \Delta_2)$.
\end{enumerate}
\end{lemma}

We also need a slightly more general substitution lemma for
transitions than the one given in Lemma~\ref{lm:rename} (1).
In the following, we denote with $n(\theta)$ the set of
all names appearing in 
the domain and range of $\theta$.
\begin{lemma}
\label{lm:trans-subst}
For any substitution $\sigma$, the following hold:
\begin{enumerate}
\item If $\one s \alpha \Delta$ and $bn(\alpha) \cap n(\sigma) = \emptyset$
then $\one{s\sigma}{\alpha\sigma}{\Delta\sigma}.$

\item If $\bigstep {\Delta}{\hat \alpha}{\Phi}$ and 
$bn(\alpha) \cap n(\sigma) = \emptyset$ 
then $\bigstep{\Delta\sigma}{\hat \alpha\sigma}{\Phi\sigma}.$
\end{enumerate}
\end{lemma}

The following lemma shows that transitions are closed under
parallel composition, under suitable conditions.

\begin{lemma}\label{lm:trans-par}
\begin{enumerate}
\item If $\one {s}{\alpha}{\Delta}$ and $fn(s') \cap bn(\alpha) = \emptyset$
then $\one{s~|~s'} {\alpha}{\Delta~|~\pdist{s'}}$
and $\one{s'~|~s}{\alpha}{\pdist{s'}~|~\Delta}.$

\item If $\bigstep {\Phi}{\hat \alpha}{\Delta}$, where $\alpha$ is either $\tau$,
a free action or a bound output, and $fn(\Phi') \cap bn(\alpha) = \emptyset$
then $\bigstep {\Phi~|~\Phi'} {\hat \alpha}{\Delta~|~\Phi'}$
and $\bigstep{\Phi'~|~\Phi}{\hat \alpha}{\Phi'~|~\Delta}.$

\item If $\one{\Phi}{a(y)}{\Phi'}$ and $\one{\Delta}{\bar a w}{\Delta'}$
then $\one{\Phi~|~\Delta}{\tau}{\Phi'[w/y]~|~\Delta'}.$

\item If $\one{\Phi}{a(y)}{\Phi'}$ and $\one{\Delta}{\bar a(y)}{\Delta'}$
then $\one{\Phi~|~\Delta}{\tau}{\nu y.(\Phi'~|~\Delta')}.$
\end{enumerate}
\end{lemma}

\begin{lemma}
\label{lm:upto-par}
If $\Rcal$ is a simulation, then $\Rcal^p \subseteq \simord$.
\end{lemma}
\begin{proof}
We show that $\Rcal^p$ is a simulation up to restriction, and
therefore, by Lemma~\ref{lm:upto-restriction}, it is included in $\simord$.

So suppose $s ~ \Rcal^p ~ \Delta$ and $\one{s}{\alpha}{\Theta}.$
By definition, we have $s = s_1 ~|~ s_2$ and $\Delta = \Delta_1 ~|~ \Delta_2$
such that $s_1 ~ \Rcal ~ \Delta_1$ and $s_2 ~\Rcal ~ \Delta_2.$

There are several cases to consider depending on the type of $\alpha$:
\begin{itemize}
\item $\alpha$ is a free output action. 
There can be two ways in which the transition $\one{s}{\alpha}{\Theta}$ is derived.
We show here one case; the other case is symmetric.
So suppose the transition is derived as follows:
$$
\infer[par]
{\one{s_1~|~s_2}{\alpha}{\Theta'~|~\pdist{s_2}}}
{
 \one{s_1}{\alpha}{\Theta'}
}
$$
where $\Theta = \Theta'~|~ \pdist{s_2}.$
Since $s_1 ~ \Rcal ~ \Delta_1$, we have 
$$
\bigstep{\Delta_1}{\hat \alpha}{\Delta_1'}
$$
and $\Theta' ~ \overline{\Rcal} ~ \Delta_1'$.
The former implies, via Lemma~\ref{lm:trans-par} (2), 
that $\bigstep{\Delta_1 ~|~ \Delta_2}{\hat\alpha}{\Delta_1' ~|~ \Delta_2}.$
Since $s_2 ~ \Rcal ~ \Delta_2$ by assumption, and therefore 
$\pdist{s_2} ~ \overline {\Rcal} ~\Delta_2$,
by Lemma~\ref{lm:par0} (2) we have 
$$\Theta = (\Theta'~|~\pdist{s_2})~\overline {\Rcal^p} ~ (\Delta_1' ~|~ \Delta_2)$$
and therefore, also 
$$\Theta = (\Theta'~|~\pdist{s_2})~\overline {\Rcal^{p \nu}} ~ (\Delta_1' ~|~ \Delta_2).$$

\item $\alpha = a(y)$ and $y \not \in fn(s,\Delta).$ That is, in this case, the transition
is derived as follows:
$$
\infer[par]
{\one{s_1 ~|~ s_2}{a(y)}{\Theta' ~|~ \pdist{s_2}}}
{
 \one{s_1}{a(y)}{\Theta'}
}
$$
and $y \not \in fn(s_2).$ 
(There is another symmetric case which we omit here.)
Since $s_1 ~ \Rcal ~\Delta_1$, we have, for every name $w$,
some $\Delta_w^1$, $\Delta_w^2$ and $\Delta_w$ such that: 
\begin{equation}
\label{eq:clo-par1}
\Delta_1 \bstep{\hat \tau} \Delta_w^1 \sstep{a(y)} \Delta_w^2, \qquad
\Delta_w^2[w/y] \bstep{\hat \tau} \Delta_w, \quad \mbox{ and }
\end{equation}
\begin{equation}
\label{eq:clo-par2}
\Theta'[w/y]~ \overline{\Rcal} ~ \Delta_w.
\end{equation}
From (\ref{eq:clo-par1}) above and Lemma~\ref{lm:trans-par} (2), and the assumption that $y \not \in fn(s,\Delta)$,
we have
$$
\Delta_1~|~\Delta_2 \bstep{\hat \tau} \Delta_w^1 ~|~ \Delta_2 
\sstep{a(y)} \Delta_w^2 ~|~ \Delta_2, \qquad
\Delta_w^2[w/y] ~|~ \Delta_2 \bstep{\hat\tau} \Delta_w ~|~ \Delta_2.
$$
Since $s_2 ~ \Rcal ~ \Delta_2$, and therefore 
$\pdist{s_2} ~ \overline{\Rcal} ~ \Delta_2$, 
it then follows from (\ref{eq:clo-par2})  and Lemma~\ref{lm:par0} (2) that 
$$
\Theta[w/y] = (\Theta'[w/y]~|~ \pdist{s_2})
~ \overline {\Rcal^p} ~ (\Delta_w ~|~ \Delta_2)
$$
and therefore
$$
\Theta[w/y] = (\Theta'[w/y]~|~ \pdist{s_2})
~ \overline {\Rcal^{p \nu}} ~ (\Delta_w ~|~ \Delta_2).
$$

\item $\alpha = \bar a(y)$ and $y \not \in fn(s,\Delta)$. This case is similar to the
previous cases, except that we only need to consider an instantiation of $y$ with a fresh
name. This is left as an exercise for the reader.

\item $\alpha = \tau$ and the transition $\one{s}{\tau}{\Theta}$ is derived via
a \textbf{Com}-rule. We show here one case; the other case can be dealt with symmetrically.
So suppose the transition is derived as follows:
$$
\infer[com]
{\one{s_1 ~|~ s_2}{\tau}{\Theta_1[w/y]~|~\Theta_2}}
{
 \one{s_1}{a(y)}{\Theta_1}
 &
 \one{s_2}{\bar a w}{\Theta_2}
}
$$
Without loss of generality, we can assume that $y \not \in fn(s,\Delta).$
Since $s_1 ~ \Rcal ~ \Delta_1$ and $s_2 ~ \Rcal ~ \Delta_2$, 
we have:
\begin{itemize}
\item For every name $w$, there are $\Lambda_1$, $\Lambda_2$ and 
$\Delta_1^w$ such that
\begin{equation}
\label{eq:clo-par3}
\Delta_1 \bstep{\hat \tau} \Lambda_1 \sstep{a(y)} \Lambda_2, \qquad
\Lambda_2[w/y] \bstep{\hat \tau} \Delta_1^w \qquad \mbox{ and }
\end{equation}
\begin{equation}
\label{eq:clo-par4}
\Theta_1[w/y] ~ \overline{\Rcal} ~ \Delta_1^w
\end{equation}

\item There exists $\Delta_2'$ such that
\begin{equation}
\label{eq:clo-par5}
\Delta_2 \bstep {\hat \tau} \Phi_1 \sstep{\bar a w} \Phi_2 \bstep{\hat \tau} \Delta_2'
\qquad \mbox{ and }
\end{equation}
\begin{equation}
\label{eq:clo-par6}
\Theta_2 ~ \overline{\Rcal} ~ \Delta_2'
\end{equation}
\end{itemize}
From (\ref{eq:clo-par3}), (\ref{eq:clo-par5}), and
Lemma~\ref{lm:trans-par} (2)-(3), we have:
$$
\Delta_1 ~|~ \Delta_2 \bstep{\hat \tau} {\Lambda_1 ~|~ \Phi_1}
\sstep {\tau} {\Lambda_2[w/y] ~|~ \Phi_2}
\bstep {\hat \tau} {\Delta_1^w ~|~ \Delta_2'},
$$
and Lemma~\ref{lm:par0} (2), together with (\ref{eq:clo-par4}) and (\ref{eq:clo-par6}),
implies
$$
(\Theta_1[w/y] ~|~ \Theta_2) ~ \overline{\Rcal^p} ~ (\Delta_1^w ~|~ \Delta_2')
$$
and therefore
$$
(\Theta_1[w/y] ~|~ \Theta_2) ~ \overline{\Rcal^{p \nu}} ~ (\Delta_1^w ~|~ \Delta_2').
$$

\item $\alpha = \tau$ and the transition $\one s \tau \Theta$ is derived 
via the \textbf{Close}-rule:
$$
\infer[close .]
{\one{s_1 ~|~ s_2}{\tau}{\nu y.(\Theta_1 ~|~ \Theta_2)}}
{
 \one{s_1}{a(y)}{\Theta_1}
 &
 \one{s_2}{\bar a(y)}{\Theta_2}
}
$$
Again, we only show one of the two symmetric cases. 
Without loss of generality, assume that $y$ is chosen to be 
fresh w.r.t. $s$ and $\Delta.$
Since $s_1 ~ \Rcal \Delta_1$ and $s_2 ~ \Rcal \Delta_2$, we have:
\begin{itemize}
\item For every name $w$, there are $\Lambda_1$, $\Lambda_2$ and $\Delta_1^w$ such that
$$
\Delta_1 \bstep{\hat \tau} \Lambda_1 \sstep{a(y)} \Lambda_2, \qquad
\Lambda_2[w/y] \bstep{\hat \tau} \Delta_1^w
\qquad \mbox{and} \qquad
\Theta_1[w/y] ~ \overline{\Rcal} ~ \Delta_1^w.
$$
Note that letting $w = y$, we have
\begin{equation}
\label{eq:clo-par7}
\Delta_1 \bstep{\hat \tau} \Lambda_1 \sstep{a(y)} \Lambda_2, \qquad
\Lambda_2 \bstep{\hat \tau} \Delta_1^y \qquad \mbox{and}
\end{equation}
\begin{equation}
\label{eq:clo-par8}
\Theta_1 ~ \overline{\Rcal} ~ \Delta_1^y
\end{equation}

\item There exist $\Phi_1$, $\Phi_2$ and  $\Delta_2'$ such that
\begin{equation}
\label{eq:clo-par9}
\Delta_2 \bstep {\hat \tau} \Phi_1 \sstep{\bar a(y)} \Phi_2 \bstep{\hat \tau} \Delta_2' \qquad \mbox{and}
\end{equation}
\begin{equation}
\label{eq:clo-par10}
\Theta_2 ~ \overline{\Rcal} ~ \Delta_2'
\end{equation}
\end{itemize}
Then, by (\ref{eq:clo-par7}), (\ref{eq:clo-par9}), 
Lemma~\ref{lm:trans-par} (2) and (4), 
and Lemma~\ref{lm:res} (1), we have: 
$$
\Delta_1 ~|~ \Delta_2 \bstep{\hat \tau} {\Lambda_1 ~|~ \Phi_1}
\sstep {\tau} {\nu y.(\Lambda_2 ~|~ \Phi_2)}
\bstep {\hat \tau} {\nu y.(\Delta_1^y ~|~ \Delta_2')}.
$$
Lemma~\ref{lm:par0} (2), together with (\ref{eq:clo-par8}) and (\ref{eq:clo-par10}),
implies
$$
(\Theta_1 ~|~ \Theta_2) ~ \overline{\Rcal^p} ~ (\Delta_1^y ~|~ \Delta_2'),
$$
which also means:
$$
(\Theta_1 ~|~ \Theta_2) ~ \overline{\Rcal^{p \nu}} ~ (\Delta_1^y ~|~ \Delta_2'). 
$$
Now by Lemma~\ref{lm:nu}, the latter implies that
$$
\nu y.(\Theta_1 ~|~ \Theta_2) ~ \overline{\Rcal^{p \nu}} ~ \nu y. (\Delta_1^y ~|~ \Delta_2'). 
$$
\end{itemize}
\qed
\end{proof}

\begin{lemma}
\label{lm:upto-par-failsim}
If $\Rcal$ is a failure simulation, then $\Rcal^p \subseteq \failsimord$.
\end{lemma}
\begin{proof}
Suppose $s \Rcal^p \Delta$ and $s\not\barb{X}$. By definition, we have
$s = s_1 ~|~ s_2$ and $\Delta = \Delta_1 ~|~ \Delta_2$ such that 
$s_1 ~ \Rcal ~ \Delta_1$ and $s_2 ~\Rcal ~ \Delta_2.$ Then we have
$s_i\not\barb{X}$ for $i=1,2$. 
Define a set $A$ as follows:
$$
A = \{ a, \bar a \mid a \in fn(s_1,s_2,\Delta_1,\Delta_2) \} \cup X. 
$$
That is, $A$ contains the set of free (co-)names in $s_i$ and $\Delta_i$ and $X.$
Let $X_i$ be the largest set such that $X \subseteq X_i \subseteq A$ and $s_i \not\barb{X_i}.$
Since $\Rcal$ is a failure simulation, it follows that there exist $\Delta_i'$ such that
$\Delta_i \dar{\tau} \Delta_i' \not \barb{X_i}.$
By Lemma~\ref{lm:trans-par} (2), we have $\Delta_1~|~\Delta_2
\dar{\tau} \Delta'_1~|~\Delta'_2.$
We claim that $(\Delta_1' ~|~ \Delta_2') \not \barb{X}.$
Suppose otherwise, that is, there exist $t_1 \in \supp {\Delta_1'}$
and $t_2 \in \supp {\Delta_2'}$ such that either 
$(t_1 ~|~ t_2) \barb \mu$, for some $\mu \in X$, or 
$(t_1 ~|~ t_2) \ar{\tau}$. If $(t_1 ~|~ t_2) \barb \mu$ then our operational semantics entails
that either $t_1 \barb \mu$ or $t_2 \barb \mu$, which contradicts the
fact that $\Delta_i' \not\barb {X_i}.$ So let's assume that 
$(t_1 ~|~ t_2) \ar{\tau}.$ Again, from the assumption $\Delta_i' \not\barb {X_i}$, we can immediately
rule out the cases where $t_i \ar{\tau}$ or $t_i \barb \mu$, for some $\mu \in X.$
This leaves us only with the cases where $t_1 \ar{\mu}$ and $t_2 \ar{\bar \mu}$
where $\mu \not \in X$ and $\bar\mu \not \in X.$ But since $\Delta_i' \not \barb{X_i}$, this
can only be the case if $\mu \not \in X_1$ and $\bar\mu \not \in X_2.$
From the operational semantics, it is easy to see that $fn(\Delta_1',\Delta_2') \subseteq fn(\Delta_1,\Delta_2)$,
so it must be the case that $\mu \in A$ and $\bar\mu \in A.$
It also must be the case that $s_1 \barb \mu$, for otherwise, it would contradict the ``largest'' property
of $X_1$. Similarly, we can argue that $s_2 \barb {\bar \mu}$. But then this would imply that
$(s_1 ~|~ s_2) \ar{\tau}$, contradicting the fact that $(s_1 ~|~ s_2) \not\barb{X}.$

The matching up of transitions and the using of $\Rcal$ to prove the
preservation property of $\failsimord$ under parallel composition are
similar to those in the corresponding proof in Lemma~\ref{lm:upto-par}
for simulations, so we omit them.
\qed
\end{proof}

\begin{lemma}\label{lm:clo-par}
\begin{enumerate}
\item
If $P_1 \simpreo Q_1$  and $P_2 \simpreo Q_2$ 
then 
$P_1 ~|~ P_2 ~ \simpreo Q_1 ~|~ Q_2.$
\item
If $P_1 \failsimpreo Q_1$  and $P_2 \failsimpreo Q_2$ 
then 
$P_1 ~|~ P_2 ~ \failsimpreo Q_1 ~|~ Q_2.$
\end{enumerate}
\end{lemma}
\begin{proof}
It is enough to show that $(\simord)^p \subseteq \simord$ and $(\failsimord)^p \subseteq \failsimord$,
which follow directly from Lemmas~\ref{lm:upto-par} and \ref{lm:upto-par-failsim} respectively. \qed
\end{proof}

\subsection{Soundness}

We now proceed to proving the main result, which is that
$P \simpreo Q$ implies $P \pmay Q$, and $P \failsimpreo Q$ implies $P \pmust Q$.
The structure of the proof follows closely that of \cite{Deng08LMCS}.
Most of the intermediate lemmas in this section are not specific
to the $\pi$-calculus; rather, they utilise the underlying probabilistic
automata semantics. 

Let $\pi^\omega$ be the set of all $\pi$ processes that may use action $\omega$. 
We write $s\ar{\alpha}_\omega\Delta$ if either $\alpha=\omega$ or
$\alpha\not=\omega$ but both $s\nar{\omega}$ and $s\ar{\alpha}\Delta$ hold. 
We define $\ar{\hat{\tau}}_w$ as we did for $\ar{\hat{\tau}}$, using $\ar{\tau}_\omega$ in place 
of $\ar{\tau}$. Similarly, we define $\dar{}_\omega$ and $\dar{\hat{\alpha}}_\omega$. 
Simulation and failure simulation are adapted to $\pi^\omega$ as follows.
\begin{definition}
Let $\failsimord^e \subseteq \pi^\omega\times
\Dcal(\pi^\omega)$ be the largest relation such that $s \failsimord^e \Theta$ implies
\begin{itemize}
\item If $ s \ar{a(x)}_\omega {\Delta}$ and $x \not \in fn(s,\Theta)$, 
then for every name $w$,  there exists $\Theta_1$, $\Theta_2$ and 
$\Theta'$ such that 
$$\Theta \dar{\hat \tau}_\omega \Theta_1 \ar{a(x)}_\omega {\Theta_2}, 
\qquad \Theta_2[w/x] \dar{\hat \tau}_\omega \Theta', 
\qquad \hbox{ and } \qquad (\Delta[w/x]) ~ \overline \Rcal ~ \Theta'. 
$$ 

\item if $s\ar{\alpha}_\omega \Delta$ and $\alpha$ is
 not an input action, then there is some $\Theta'$
  with $\Theta\dar{\hat{\alpha}}_\omega\Theta'$ and
  $\Delta\lift{\failsimord^e}\Theta'$

\item if $s\not \barb{X}$ with $\omega\in X$ then there is some 
$\Theta'$ with $\Theta\dar{\hat{\tau}}_\omega\Theta'$ and $\Theta'\not \barb{X}$.

\end{itemize} 
Similarly we can define $\simord^e$ by dropping the third clause.
Let $P \failsimpreo^e Q$ if $\interp{P}\dar{\hat{\tau}}_\omega\Theta$ for some $\Theta$ with $\interp{Q}\lift{\failsimord^e}\Theta$.  Similarly, $P\simpreo^e Q$ if $\interp{Q}\dar{\hat{\tau}}_\omega\Theta$ for some $\Theta$ with $\interp{P}\lift{\simord^e}\Theta$. 
\end{definition}
Note that for $\pi$-processes $P,Q$, there is no action $\omega$, therefore we have $P\failsimpreo Q$ iff $P\failsimpreo^e Q$, and $P\simpreo Q$ iff $P\simpreo^e Q$.

\begin{lemma}\label{lem:preserve.par}
Let $P,Q$ be processes in $\pi$ and $T$ be a process in $\pi^\omega$.
\begin{enumerate}
\item If $P\simpreo Q$ then $T~|~ P \simpreo^e T ~|~Q$.
\item If $P\failsimpreo Q$ then $T~|~ P \failsimpreo^e T ~|~Q$.
\end{enumerate}
\end{lemma}
\begin{proof}
Similar to the proof of Lemma~\ref{lm:clo-par}.
\qed
\end{proof}

\begin{lemma}
\label{lem:pmay-max}
\begin{enumerate}
\item
$P \pmay Q$ if and only if for every test $T$ we have
$$
max(\val(\interp {\nu \vec x.(T ~|~ P)})) \leq
max(\val(\interp {\nu \vec x.(T ~|~ Q)}))
$$
where $\vec x$ contain the free names of $T$, $P$ and $Q$, excluding $\omega.$

\item
$P \pmust Q$ if and only if for every test $T$ we have
$$
min(\val(\interp {\nu \vec x.(T ~|~ P)})) \leq
min(\val(\interp {\nu \vec x.(T ~|~ Q)}))
$$
where $\vec x$ contain the free names of $T$, $P$ and $Q$, excluding $\omega.$
\end{enumerate}
\end{lemma}
\begin{proof}
The results follow from the simple fact that, for non-empty finite outcome sets $O_1,O_2$,
\begin{itemize}
\item $O_1\sqsubseteq_{Ho}O_2$ iff $max(O_1)\leq max(O_2)$
\item $O_1\sqsubseteq_{Sm}O_2$ iff $min(O_1)\leq min(O_2)$
\end{itemize}
which is established as Proposition 2.1 in \cite{Deng07ENTCS}.
\qed
\end{proof}

\begin{lemma}
\label{lm:trans-max}
$\Delta_1 \dar{\hat \tau} \Delta_2$ implies
$max(\val(\Delta_1)) \geq max(\val(\Delta_2))$ and $min(\val(\Delta_1)) \leq min(\val(\Delta_2))$.
\end{lemma}
\begin{proof}
Similar properties are proven in \cite[Lemma 6.15]{Deng07ENTCS} using a function $maxlive$ instead of $max\circ\val$. Essentially the same arguments apply here.
\qed
\end{proof}

\begin{proposition}
\label{prop:sim-max}
\begin{enumerate}
\item $\Delta_1\lift{\simord^e}\Delta_2$ implies $max(\val(\Delta_1)) \leq max(\val(\Delta_2))$.
\item $\Delta_1\lift{\failsimord^e}\Delta_2$ implies $min(\val(\Delta_1)) \geq min(\val(\Delta_2))$.
\end{enumerate}
\end{proposition}
\begin{proof}
The first clause is proven in \cite[Proposition 6.16]{Deng07ENTCS} using a function $maxlive$ instead of $max\circ\val$. The second clause is proven in \cite[Proposition 4.10]{Deng08LMCS}
\qed
\end{proof}

\begin{theorem}\label{thm:sim-sound}
\begin{enumerate}
\item
$P \simpreo Q$ implies $P \pmay Q$
\item
$P \failsimpreo Q$ implies $P \pmust Q.$
\end{enumerate}
\end{theorem}
\begin{proof}
We prove the second statement; similar is the first one.
Suppose $P \failsimpreo Q$. Given Proposition~\ref{lem:pmay-max},
it is sufficient to show that for every test $T$,
$$
min(\val(\interp{\nu \vec x (T ~|~ P) })) \leq min(\val(\interp{\nu \vec x(T ~|~ Q)}))
$$
where $\vec x$ contain the free names of $T$, $P$ and $Q$, but excluding $\omega.$
Since $\failsimpreo$ is preserved by parallel composition (cf. Lemma~\ref{lem:preserve.par}) and name restriction, we have
that
$$
\nu \vec x(T ~|~ P) \failsimpreo^e \nu \vec x(T~|~Q),
$$
which means there is a $\Theta$  such that
$\interp {\nu \vec x(T~|~P)} \bstep{\hat \tau} \Theta$ and
$\interp {\nu \vec x(T~|~Q)} ~ \lift{\failsimord^e} \Theta.$
The result then follows from Proposition~\ref{prop:sim-max} and
Lemma~\ref{lm:trans-max}.
\qed 
\end{proof}

\section{A modal logic for $\probpi$}
\label{sec:modal}

We consider a modal logic based on a fragment of
Milner-Parrow-Walker's (MPW) modal logic for the (non-probabilistic)
$\pi$-calculus~\cite{Milner93TCS}, but extended with a probabilistic
disjunction operator $\oplus$, similar to that used in \cite{Deng08LMCS}.
The language of formulas is given by the following grammar:
$$
\varphi ::= \top ~ \mid ~ \Ref{X} ~ \mid ~ \ldia{a(x)} \varphi ~ \mid ~ 
            \ldia{\bar a x} \varphi ~ \mid ~
            \ldia{\bar a(x)} \varphi ~  \mid ~ 
            \varphi_1 \wedge \varphi_2 ~ \mid ~
\varphi_1 {\pch p} \varphi_2
$$
The $x$'s in $\ldia{a(x)}\varphi$ and $\ldia{\bar a(x)}\varphi$ are 
binders, whose scope is over $\varphi.$
The diamond operator $\ldia{a(x)}$ is called a bound input modal operator,
$\ldia{\bar a x}$ a free output modal operator and $\ldia{\bar a(x)}$ a bound output
modal operator.
Instead of binary conjunction and probabilistic disjunction, we sometimes write $\pand_{i\in I}\varphi_i$ and $\varphi_1 {\pch p} \varphi_2$ for finite index set $I$; they can be expressed by nested use of their binary forms. 
We refer to this modal logic as $\Fcal$.
Let $\Lcal$ be the sub-logic of $\Fcal$ by skipping the $\Ref{X}$ clause.  
The semantics of each operator is defined as follows.

\begin{definition}
\label{def:sat}
The {\em satisfaction relation} $\models$ between a distribution and 
a modal formula is defined inductively as follows:
\begin{itemize}
\item $\Delta \models \top$ always. 
\item $\Delta \models \Ref{X}$ iff there is a $\Delta'$ with $\Delta\dar{\hat{\tau}} \Delta'$ and $\Delta'\not\barb{X}$.
\item $\Delta \models \ldia{a(x)}\varphi$ iff 
for all $z$ there are $\Delta_1,$ $\Delta_2,$ $\Delta'$ and $w$ such that 
$\Delta \bstep{\hat \tau} \Delta_1 \sstep{a(w)} \Delta_2$,
$\Delta_2[z/w] \bstep{\hat \tau} \Delta'$
and
$\Delta' \models \varphi[z/x].$

\item $\Delta \models \ldia{\bar a x}\varphi$ iff for some $\Delta'$,
$\Delta \bstep{\widehat{\bar a x}} \Delta'$
and $\Delta' \models \varphi.$

\item $\Delta \models \ldia{\bar a(x)}\varphi$ iff for some
$\Delta'$ and $w \not \in fn(\varphi, \Delta)$,
$\Delta \bstep{\widehat{\bar a(w)}} \Delta'$
and $\Delta' \models \varphi[w/x].$

\item $\Delta \models \varphi_1 \wedge \varphi_2$ iff
$\Delta \models \varphi_1$ and $\Delta \models \varphi_2$.

\item $\Delta \models \varphi_1 {\pch p} \varphi_2$ iff there are $\Delta_1,\Delta_2 \in \Dcal(S_p)$ with $\Delta_1\models\varphi_1$ and $\Delta_2\models\varphi_2$, such that
$\bigstep{\Delta}{\hat \tau}{p \cdot \Delta_1 + (1-p)\cdot\Delta_2.}$
\end{itemize}
We write $\Delta \lleq \Theta$ just when 
$\Delta \models \psi$ implies $\Theta \models \psi$ for
all $\psi \in \Lcal$, and $\Delta \fleq \Theta$ just when $\Theta \models \varphi$ implies $\Delta \models \varphi$ for all $\varphi\in\Fcal$.
We write $P \lleq Q$
when $\interp P \lleq \interp Q$, and $P \fleq Q$ when $\interp P \fleq \interp Q$.
\end{definition}

Following \cite{Deng08LMCS}, in order to show soundness 
of the logical preorders w.r.t. the simulation pre-orders, we need to define a notion
of characteristic formulas. 

\begin{definition}[Characteristic formula]
\label{def:char-form}
The {\em $\Fcal$-characteristic formulas} $\varphi_s$ and $\varphi_\Delta$
of, respectively, a state-based process $s$ and a distribution $\Delta$
are defined inductively as follows:
$$
\begin{array}{rcl}
\varphi_s & := & \pand\{\ldia{\alpha}\varphi_\Delta \mid s\ar{\alpha}\Delta\} \wedge \Ref{\{\mu \mid s\not\barb{\mu}\}} \qquad\mbox{ if $s\not\ar{\tau}$},\\
\varphi_s & := & \pand\{\ldia{\alpha}\varphi_\Delta \mid s\ar{\alpha}\Delta, ~ \alpha \not = \tau\} ~ \land ~
\pand\{\varphi_\Delta \mid s\ar{\tau}\Delta\} \qquad
\mbox{ otherwise}.\\
\varphi_\Delta & := & \psum_{s\in\supp{\Delta}}\Delta(s) \cdot \varphi_s
\end{array}
$$
where $\psum$ is a generalised probabilistic choice as in Section~\ref{sec:pi}.
The \emph{$\Lcal$-characteristic formulas} $\psi_s$ and $\psi_\Delta$ are defined likewise, but omitting the conjuncts $\Ref{\{\mu \mid s\not\barb{\mu}\}}$.
\end{definition}
Note that because we use the late semantics (cf. Figure~\ref{fig:pi}), the conjunction in $\varphi_s$ is finite even though there can
be infinitely many (input) transitions from $s.$

Given a state based process $s$, we define its {\em size}, $|s|$, as the
number of process constructors and names in $s.$
The following lemma is straightforward from the definition of the operational
semantics of $\pi_p$.

\begin{lemma}
\label{lm:trans-size}
If $s \sstep{\alpha} \Delta$ then $|s| > |t|$ for every $t \in \supp \Delta.$
\end{lemma}

\begin{lemma}
\label{lm:char-form}
For every $\Delta \in \Dcal(S_p)$, $\Delta \models \varphi_\Delta$, as well as $\Delta\models\psi_\Delta$.
\end{lemma}
\begin{proof}
It is enough to show that $\bar s \models \varphi_s.$ 
This is proved by by induction on $|s|.$ 
So suppose $s\not\ar{\tau}$. Then we have
$$
\begin{array}{ll}
\varphi_s = & \Ref{\{\mu \mid s\not\barb{\mu}\}}\wedge \\
& \pand \{\ldia {a(x)} \varphi_\Delta \mid \one s {a(x)}{\Delta} \} \land  
 \pand \{\varphi_\Delta \mid \one s \tau \Delta \} \land \\
 & \pand \{\ldia {\bar a x } \varphi_\Delta \mid \one s {\bar a x }{\Delta} \} \land  
 \pand \{\ldia {\bar a(x) } \varphi_\Delta \mid \one s {\bar a(x) }{\Delta} \}. 
\end{array}
$$
where $
\varphi_\Delta = \psum_{s \in \supp \Delta} \Delta(s).\varphi_s.
$
For each of the conjunct $\phi$, we prove that $\pdist s \models \phi.$
We show here two cases; the other cases are similar.
\begin{itemize}
\item $\phi= \Ref{X}$, where $X=\{\mu \mid s\not\barb{\mu}\}$. 
For each $\mu \in X$ we have $s\not\barb{\mu}$. Moreover, since $s\not\ar{\tau}$, 
we see that $s\not\barb{X}$.

\item  $\phi = \ldia {a(x)} \varphi_\Delta$. 
So suppose $s \sstep{a(x)} \Delta$ and 
$\supp \Delta = \{s_i \mid i \in I\}$ and $\Delta = \sum_{i \in I} p_i \cdot \pdist{s_i}.$
Since $|s_i| < |s|$, by the induction hypothesis, 
for every name $w$, we have 
$$
\pdist{s_i[w/x]} \models \varphi_{s_i[w/x]}
$$
and therefore:
$$
\Delta[w/x] = \sum_{i\in I} p_i \cdot \pdist{s_i[w/x]} \models 
\psum_{i \in I} p_i \cdot \varphi_{s_i[w/x]} = \varphi_\Delta[w/x].
$$
Let $\Phi_1 = \Phi_2 = \pdist s.$ 
Obviously we have, for every $w$, 
$$
\Phi_1 \bstep {\hat \tau} \Phi_2 \sstep {a(x)} \Delta, \qquad
\Delta[w/x] \models \varphi_\Delta[w/x].
$$
So by Definition~\ref{def:sat}, $\pdist s \models \phi.$
\end{itemize}
\qed
\end{proof}

\begin{lemma}
\label{lm:modal-sim}
For any processes $P$ and $Q$,
$\interp{P}\models\varphi_{\interp{Q}}$ implies $P\failsimpreo Q$, and likewise
 $\interp Q \models \psi_{\interp P}$
implies $P \simrel_S Q.$
\end{lemma}
\begin{proof}
Let $\Rcal$ be the relation defined as follows:
$s ~ \Rcal ~ \Theta$ iff $\Theta \models \varphi_s.$
We first prove the following claim:
\begin{equation}
\label{claim}
\mbox{
$\Theta \models \varphi_\Delta$ implies
there exists $\Theta'$ such that
$\bigstep{\Theta}{\hat \tau}{\Theta'}$
and $\Delta ~ \overline \Rcal ~ \Theta'.$
}
\end{equation}
To prove this claim (following \cite{Deng08LMCS}), suppose that $\Theta \models \Delta$.
By definition, $\varphi_\Delta = \psum_{i \in I} p_i \cdot \varphi_{s_i}$
and $\Delta = \sum_{i \in I} p_i \cdot \pdist {s_i}$. For every $i \in I$,
we have $\Theta_i \in \Dcal(S_p)$ with $\Theta_i \models \varphi_{s_i}$
such that $\bigstep {\Theta} {\hat \tau}{\Theta'}$ with 
$\Theta' = \sum_{i\in I} p_i \cdot \Theta_i.$ Since $s_i ~ \Rcal ~ \Theta_i$
for all $i \in I$, we have $\Delta ~ \overline \Rcal ~ \Theta'.$

We now proceed to show that $\Rcal$ is a failure simulation, hence proving
the first statement of the lemma. So suppose $s ~ \Rcal ~ \Theta$.
\begin{enumerate}
\item Suppose $\one s \tau \Delta$. By the definition of $\Rcal$,
we have $\Theta \models \varphi_s.$
By Definition~\ref{def:char-form}, we also have $\Theta \models \varphi_\Delta.$
By (\ref{claim}) above, there exists $\Theta'$
such that $\bigstep{\Theta}{\hat \tau}{\Theta'}$ and $\Delta ~ \overline \Rcal ~ \Theta'.$

\item Suppose $\one s {\bar a x} \Delta$.
Then by Definition~\ref{def:char-form}, $\Theta \models \ldia{\bar a x}\varphi_\Delta.$
So $\Theta \bstep {\bar a x } \Theta'$ and $\Theta' \models \varphi_\Delta$,
for some $\Theta'.$ By (\ref{claim}), there exists $\Theta''$ such that
$\Theta'  \bstep{\hat {\tau}}{\Theta''}$ and $\Delta ~ \overline \Rcal ~ \Theta''.$
This means that $\Theta \bstep{\bar a x} \Theta''$ and
$\Delta ~ \overline \Rcal ~ \Theta''.$

\item Suppose $\one s {a(x)} \Delta$ for some $x \not \in fn(s,\Theta).$
By Definition~\ref{def:char-form}, $\Theta \models \ldia{a(x)}\varphi_\Delta.$
This means for every name $z$, there exists $\Theta_z^1$, $\Theta_z^2$ and $\Theta_z$ such that
$\Theta \bstep {\hat \tau} \Theta_z^1 \sstep {a(x)} {\Theta_z^2}$,
$\Theta_z^2[z/x] \bstep{\hat \tau} \Theta_z$ and $\Theta_z \models \varphi_\Delta[z/x].$\footnote{Strictly
speaking, we should also consider the case where $\Theta_z^1 \sstep{a(w)} \Theta_z^2$,
but it is easy to see that since $x \not \in fn(s,\Theta)$ we can always apply a renaming to rename
$w$ to $x.$}
Then by (\ref{claim}) we have 
$\bigstep {\Theta_z }{\hat \tau}{\Theta_z'}$
and $\Delta[z/x] ~ \overline \Rcal ~ \Theta_z'.$
So we indeed have, for every name $z$, $\Theta_z^1$, $\Theta_z^2$ and $\Theta_z'$ such that
$$
\Theta \bstep{\hat \tau} \Theta_z^1 \sstep{a(x)} \Theta_z^2,
\qquad 
\Theta_z^2[z/x] \bstep{\hat \tau} \Theta_z' 
\qquad
\hbox{ and }
\qquad
\Delta[z/x] ~ \overline \Rcal ~ \Theta_z'.
$$

\item Suppose $\one s {\bar a(x)}{\Delta}.$ This case is similar
to the previous one, except that we need only to consider one
instance of $x$ with a fresh name.

\item Suppose $s\not\barb{X}$ for a set of channel names $X$. By Definition~\ref{def:char-form}, we have $\Theta\models\Ref{X}$. Hence, there is some $\Theta'$ with $\Theta\dar{\hat{\tau}}\Theta'$ and $\Theta'\not\barb{X}$.
\end{enumerate}

To establish the second statement, define $\Rcal$ by $s\Rcal\Theta$ iff $\Theta\models\psi_s$. Just as above it can be shown that $\Rcal$ is a simulation. Then the second statement of the lemma easily follows.
\qed
\end{proof}

\begin{theorem}\label{thm:modal-sim}
\begin{enumerate}
\item
If $P \lleq Q$ then $P \simpreo Q.$
\item
If $P \fleq Q$ then $P \failsimpreo Q.$
\end{enumerate}
\end{theorem}
\begin{proof}
Suppose $P \simrel_\Lcal Q$. By Lemma~\ref{lm:char-form}, we have $\interp P \models \psi_{\interp P}$,
hence $\interp Q \models \psi_{\interp P}$. Then by Lemma~\ref{lm:modal-sim}, we have $P \simrel_S Q.$

For the second statement, assume $P \failsimpreo Q$, we have $\interp{Q}\models\varphi_{\interp{Q}}$ and hence $\interp{P}\models\varphi_{\interp{Q}}$, and thus $P\failsimpreo Q$.
\qed
\end{proof}

\newcommand\vsim[1]{{\widehat \simrel}^{#1}_{pmay}}
\newcommand\vapply[2]{{\widehat {\cal A}}^\Omega_{\updownarrow}(#1,#2) }

\newcommand\pr[2]{\langle #1, #2\rangle}

\section{Completeness of the simulation preorders}
\label{sec:comp}

In the following, we assume a function $new$ that takes as an argument 
a finite set of names and outputs a fresh name, i.e.,
if $new(N) = x$ then $x\not \in N.$
If $N = \{x_1,\ldots,x_n\}$, we write $[x \not = N]P$ to abbreviate
$[x \not = x_1][x \not = x_2] \cdots [x \not = x_n]P.$

For convenience of presentation, we write $\vec{\omega}$ for the vector in $[0,1]^\Omega$
defined by $\vec{\omega}(\omega)=1$ and $\vec{\omega}(\omega')=0$ for
any $\omega'\not=\omega$. 
We also extend the $Apply^\Omega$ function to
allow applying a test to a distribution, defined as 
$Apply^\Omega(T, \Delta) = \val({\nu \vec x(\interp T~|~\Delta)})$
where $\vec x = fn(T,\Delta) - \Omega.$ 

\begin{lemma}
\label{lm:sat-renaming}
If $\Delta \models \varphi$ then $\Delta \sigma \models \varphi \sigma$
for any renaming substitution $\sigma.$
\end{lemma}

In the following, given a name $a$, we write $a.P$ to denote $a(y).P$ for
some $y \not \in fn(P).$ Similarly, we write $\bar a.P$ to denote 
$\bar a a.P.$
Recall that the size of a state-based process, $|s|$, is the number of
symbols in $s.$ The {\em size} of a distribution $\Delta$, written $|\Delta|$, is 
the {\em multiset} $\{|s| \mid s \in \supp \Delta \}.$
There is a well-founded ordering on $|\Delta|$, i.e., the multiset (of natural
numbers) ordering, which we shall denote with $\prec$.

\begin{lemma}
\label{lm:test1}
Let $P$ be a process and $T, T_i$ be tests. 
\begin{enumerate}
\item $o\in Apply^\Omega(\omega,P)$ iff $o=\vec{\omega}$.
\item Let $X=\{\mu_1,...,\mu_n\}$ and $T=\mu_1.\omega+...+\mu_n.\omega$. Then 
 $\vec{0} \in Apply^\Omega(T,P)$ iff
 $\interp{P}\dar{\hat{\tau}}\Delta$ for some $\Delta$ with $\Delta\not\barb{X}$.
\item Suppose the action $\omega$ does not occur in the test $T$. Then
$o\in Apply^\Omega(\omega+a(x).([x=y]\tau.T+\omega), P) $ with $o(\omega)=0$ iff
there is $\Delta$ such that 
$\interp P \bstep {\widehat{\bar a y}} \Delta$
and $o\in Apply^\Omega(T[y/x], \Delta).$
\item Suppose the action $\omega$ does not occur in the test $T$ and  $fn(P)\subseteq N$. Then
$o\in Apply^\Omega(\omega+a(x).([x\not=N]\tau.T+\omega), P) $ with $o(\omega)=0$ iff
there is $\Delta$ such that 
$\interp P \bstep {\widehat{\bar a (y)}} \Delta$
and $o\in Apply^\Omega(T[y/x], \Delta).$
\item Suppose the action $\omega$ does not occur in the test $T$. Then
$o\in Apply^\Omega(\omega+\bar a x.T, P)$ with $o(\omega)=0$
iff there are $\Delta$, $\Delta_1$ and $\Delta_2$ 
such that $\interp P \bstep{\hat \tau} \Delta_1 \sstep{a(y)} \Delta_2,$ 
$\Delta_2[x/y] \bstep{\hat \tau} \Delta$ 
and 
$o\in Apply^\Omega(T,\Delta).$
\item $o\in Apply^\Omega(\psum_{i\in I} p_i \cdot T_i, P) $
iff $o=\sum_{i\in I}p_i\cdot o_i$ for some
$o_i\in Apply^\Omega(T_i,P)$ for all $i\in I.$

\item $o\in Apply^\Omega(\sum_{i\in I} \tau.T_i, P)$
if for all $i \in I$ there are $q_i \in [0,1]$ and
$\Delta_i$ such that $\sum_{i\in I} q_i = 1$,
$\interp P \bstep{\hat \tau} \sum_{i \in I} q_i \cdot \Delta_i$
and $o=\sum_{i\in I}q_i\cdot o_i$ for some
$o_i\in Apply^\Omega(T_i,\Delta_i).$
\end{enumerate}
\end{lemma}
\begin{proof}
The proofs of items 1 and 2 are similar to the proofs of Lemma 6.7(1) and 6.7(2) 
in \cite{Deng08LMCS} for pCSP; items 6 and 7 correspond to 
Lemma 6.7(4) and Lemma 6.7(5) in \cite{Deng08LMCS}, respectively. 
Items 3, 4 and 5 have a counterpart in Lemma 6.7(3) of \cite{Deng08LMCS}, but they
are quite different, due to the name-passing feature of the $\pi$-calculus, and the
possibility of checking the identity of the input value via the match and the mismatch
operators. 
We show here a proof of item 3; the proofs of items 4 and 5 are similar.

We first generalize item 3 to distributions: given $\omega$ and $T$ as above, we have,
for every distribution $\Theta$, 
\begin{quote}
$o\in Apply^\Omega(\omega+a(x).([x=y]\tau.T+\omega), \Theta) $ with $o(\omega)=0$ iff
there is $\Delta$ such that 
$\Theta \bstep {\widehat{\bar a y}} \Delta$
and $o\in Apply^\Omega(T[y/x], \Delta).$
\end{quote}
The `if' part is straightforward from Definition~\ref{def:vector-based-results}. 
We show the `only if' part here. 
The proof will make use of the following claim (easily proved by induction on $|\Theta|$):
\begin{equation}
\label{eq:claim}
\begin{array}{l}
\mbox{\bf Claim:} ~ o\in Apply^\Omega([y=y]\tau.T[y/x]+\omega, \Theta) \mbox{ with } o(\omega)=0 \mbox{ iff } \\
\mbox{there is $\Delta$ such that $\Theta \bstep {\hat \tau} \Delta$ and $o\in Apply^\Omega(T[y/x], \Delta).$ }
\end{array}
\end{equation}

So, suppose we have $o\in Apply^\Omega(\omega+a(x).([x=y]\tau.T+\omega), \Theta) $ with $o(\omega)=0$.
We show, by induction on $|\Theta|$,  that there exists $\Delta$ such that 
$\Theta \bstep {\bar a y} \Delta$
and $o\in Apply^\Omega(T[y/x], \Delta).$
Let $T' = \omega+a(x).([x=y]\tau.T+\omega)$, and suppose
$\Theta = p_1 \cdot \pdist{s_1} + \ldots + p_n \cdot \pdist{s_n}$, 
for pairwise distincts state-based processes $s_1,\ldots,s_n$,
and suppose that 
$\vec z$ is an enumeration of the set $fn(T',\Theta) - \Omega.$
Then 
$$
Apply^\Omega(T',\Theta) = \val^\Omega(p_1 \cdot \pdist{\nu \vec z(T'|s_1)} + \ldots + p_n \cdot \pdist{\nu \vec z(T'|s_n)}).
$$
From Definition~\ref{def:vector-based-results}, in order to have $o(\omega) = 0$,
it must be the case that $\nu \vec z(T' | s_j) \ar{\tau} $ for 
every $j \in \{1,\dots,n\}.$
From the definition of the operational semantics, there are exactly two cases where this might happen:
\begin{itemize}
\item For some $i$, $s_i \ar{\tau} \Lambda$ for some distribution $\Lambda.$ 
Let $\Theta' = p_1 \cdot \pdist{s_1} + \ldots + p_i \cdot \Lambda + \ldots + p_n \cdot \pdist{s_n}.$
Then we have $\Theta \ar{\hat \tau} \Theta'$ and 
$\nu\vec z(T' | \Theta) \ar{\hat \tau} \nu \vec z(T'|\Theta').$
The latter means that $o \in \val^\Omega(\nu \vec z(T'|\Theta'))$ as well.
By Lemma~\ref{lm:trans-size}, we know that $|\Lambda| \prec \{|s_i|\}$,
and therefore $|\Theta'| \prec |\Theta|.$ By the induction hypothesis, 
$$
\Theta \ar{\hat\tau} \Theta' \bstep{\widehat{\bar a y}} \Delta
$$
and $o \in Apply^\Omega(T[y/x], \Delta).$

\item For every $i \in \{1,\dots,n\}$, we have $s_i \not \ar{\tau}.$
This can only mean that the $\tau$ transition from $\nu \vec z(T'|s_i)$
derives from a communiation between $T'$ and $s_i.$
This means that $s_i\barb{\bar a}$, for every $i \in \{1,\dots,n\}.$
We claim that, in fact, for every $i$, we have $s_i \ar{\bar a y} \Theta_i$,
for some $\Theta_i.$ For otherwise, we would have that for some
$j$, $\nu \vec z(T' | s_j) \ar{\tau} \nu \vec z(([u = y]\tau.T[y/x] + \omega) ~|~ \Theta_j)$,
for some $u$ distinct from $y.$ But this means that only the $\omega$ action is enabled
in the test, so all results of $\val^\Omega(\nu \vec z(([u = y]\tau.T[y/x] + \omega) ~|~ \Theta_i))$ 
in this case would have a non-zero $\omega$ component, which would mean that 
$o(\omega)$ would be non-zero as well, contradicting the assumption that $o(\omega) = 0$.
So, we have $s_i \ar{\bar a y} \Theta_i$ for every $i \in \{1,\dots,n\}.$
Let $\Theta'=p_1 \cdot \Theta_1 + \ldots + p_n \cdot \Theta_n.$
Then we have $\Theta \ar{\bar a y} \Theta'$ and
$\nu \vec z(T' ~|~ \Theta) \ar{\tau} \nu \vec z(T'' ~|~ \Theta')$
where $T'' = [y=y]\tau.T[y/x] + \omega$.
The latter transition means that $o \in \val^\Omega(\nu \vec z(T'' ~|~ \Theta')) = Apply^\Omega(T'',\Theta').$ 
We can therefore apply Claim~\ref{eq:claim} to get:
$$
\Theta \ar{\bar a y} \Theta' \bstep{\hat \tau} \Delta
$$
and $o \in Apply^\Omega(T[y/x], \Delta).$
\end{itemize}
\qed
\end{proof}

\begin{lemma}
\label{lm:test2}
If $o\in Apply^\Omega(\sum_{i\in I} \tau.T_i, P)$ then for all
$i \in I$ there are $q_i \in [0,1]$ and $\Delta_i$ with $\sum_{i\in I} q_i = 1$
such that $\interp P \bstep {\hat \tau} \sum_{i\in I} q_i \cdot \Delta_i$
and $o=\sum_{i\in I}q_i\cdot o_i$ for some $o_i\in Apply^\Omega(T_i,\Delta_i).$
\end{lemma}
\begin{proof}
The proof is similar to the proof of Lemma 6.8 in \cite{Deng08LMCS}.
\qed
\end{proof}

The key to the completeness proof is to find a `characteristic test' for
every formula $\varphi \in \Lcal$ with a certain property.
The construction of these characteristic tests is given in the following lemma. 
Note that unlike in the case of pCSP~\cite{Deng08LMCS}, 
this construction is parameterised by a finite set of names $N$, representing
the set of free names of the process/distribution on which the test applies to. 
This parameter is important for the test to be able to detect output of fresh names. 

\begin{lemma}\label{lm:comp}
For every finite set of names $N$ and every $\varphi \in \Fcal$ such that
$fn(\varphi) \subseteq N$, there exists  a test $T_{\pr N \varphi}$ and $v_\varphi \in [0,1]^\Omega$,
such that
\begin{equation}
\label{eq:comp-ex1}
\Delta \models \varphi \qquad \hbox{ iff } \qquad
 \exists o\in Apply^\Omega(T_{\pr N \varphi}, \Delta): o\leq v_\varphi
\end{equation}
for every $\Delta$ with $fn(\Delta) \subseteq N$, and in case $\varphi\in\Lcal$ we also have
\begin{equation}
\label{eq:comp-ex2}
\Delta \models \varphi \qquad \hbox{ iff } \qquad
 \exists o\in Apply^\Omega(T_{\pr N \varphi}, \Delta): o\geq v_\varphi.
\end{equation}
$T_{\pr N \varphi}$ is called a \emph{characteristic test} of $\varphi$ and $v_\varphi$ its \emph{target value}.
\end{lemma}
\begin{proof}
The characteristic tests and target values are defined by induction on $\varphi$:
\begin{itemize}
\item $\varphi = \top$: Let $T_{\pr N \varphi} := \omega$ for some $\omega\in\Omega$ and $v_\varphi:=\vec{\omega}$.

\item $\varphi = \Ref{X}$ with $X=\{\mu_1,...,\mu_n\}$. 
Let $T_\varphi:=\mu_1.\omega+...+\mu_n.\omega$ for some $\omega\in\Omega$, and $v_\varphi=\vec{0}$.

\item $\varphi = \ldia{\bar a x} \psi$: Let
$T_{\pr N \varphi} :=\omega+ a(y).([y=x]\tau.T_{\pr N \psi}+ \omega)$
for some $y \not \in fn(T_{\pr N \psi})$, where $\omega\in\Omega$ does not occur in $T_{\pr N \psi}$ and $v_\varphi:=v_\psi$. 

\item $\varphi = \ldia{\bar a (x)}\psi$:  
Let $z = new(N)$ and $N' = N \cup \{z\}.$ Without loss of generality,
we can assume that $x = z$ (since we consider terms equivalent modulo
$\alpha$-conversion). 
Then let
$
T_{\pr N \varphi} := \omega+a(x).([x \not = N] \tau.T_{\pr {N'} \psi}+\omega)
$, where $\omega\in\Omega$ does not occur in $T_{\pr {N'} \psi}$ and $v_\varphi:=v_\psi$. 

\item $\varphi = \ldia{a(x)}\psi$:
Let $z = new(N)$ and $N' = N \cup \{z\}.$
Let $p_w \in (0,1]$ for $w \in N'$ be chosen arbitrarily such that
$\sum_{w \in N'} p_w = 1.$
Then let
$$
T_{\pr N \varphi} := 
\psum_{w \in N'} p_w \cdot (\omega_w+\bar a w.T_{\pr {N'} {\psi[w/x]}})
$$
where $\omega_w$ does not occur in $T_{\pr {N'} {\psi[w/x]}}$ for each
$w\in N'$, and $\omega_{w_1}\not=\omega_{w_2}$ if $w_1\not=w_2$. We let
 $v_\varphi:= \sum_{w\in N'}p_w\cdot v_{\psi[w/x]}$.

\item $\varphi = \pand_{i \in I} \varphi_i$ where $I$ is a finite and non-empty
index set. 
Choose an $\Omega$-disjoint family $(T_{\pr
  N {\varphi_i}}, v_{\varphi_i})_{i\in I}$ of characteristic tests and
target values. Let $p_i \in (0,1]$ for $i \in I$ be chose arbitrarily such that
$\sum_{i \in I} p_i = 1.$
Then let
$$T_{\pr N \varphi} := \psum_{i\in I} p_i \cdot T_{\pr N {\varphi_i}}$$
and $v_\varphi:=\sum_{i\in I}p_i\cdot v_{\varphi_i}$.

\item $\varphi = \psum_{i\in I} p_i \cdot \varphi_i.$ Choose an $\Omega$-disjoint family $(T_i,v_i)_{i\in I}$ of characteristic tests $T_i$ with target values $v_i$ for each $\varphi_i$, such that there are distinct success actions $\omega_i$ for $i\in I$ that do not occur in any of those tests. Let $T'_i:=T_i\pch{\frac{1}{2}}\omega_i$ and $v'_i:=\frac{1}{2}v_i+\frac{1}{2}\vec{\omega_i}$. Note that for all $i\in I$ also $T'_i$ is a characteristic test of $\varphi_i$ with target value $v'_i$. Let
 $T_{\pr N \varphi} := \sum_{i\in I} \tau.T_{\pr N {\varphi_i}}$ and $v_\varphi:=\sum_{i\in I}p_i\cdot v'_i$.
\end{itemize}

We now prove (\ref{eq:comp-ex1}) above by induction on $\varphi$: 
\begin{itemize}
\item $\varphi = \top$: obvious.

\item $\varphi = \Ref{X}$. Suppose $\Delta\models\varphi$. Then there is a $\Delta'$ with $\Delta\dar{\hat{\tau}}\Delta'$ and $\Delta'\not\barb{X}$. By Lemma~\ref{lm:test1}(2), $\vec{0}\in Apply^\Omega(T_{\pr N \varphi},\Delta)$.

Now suppose $\exists o\in Apply^\Omega(T_{\pr N \varphi},\Delta): o\leq v_\varphi$. This means $o=\vec{0}$, so by Lemma~\ref{lm:test1}(2) there is a $\Delta'$ with $\Delta\dar{\tau}\Delta'$ and $\Delta'\not\barb{X}$. Hence $\Delta\models\varphi$.

\item $\varphi = \ldia {\bar a x} \phi:$
Suppose $\Delta \models \varphi.$
Then $\Delta \bstep{\bar a x } \Delta'$ and $\Delta' \models \phi.$
By the induction hypothesis, $\exists o\in Apply^\Omega(T_{\pr N 
  \phi}, \Delta'): o\leq v_\phi$.
By Lemma~\ref{lm:test1}(3), this means $o\in
Apply^\Omega(\omega+a(y).([y=x]\tau.T_{\pr N \phi}+\omega), \Delta)$. Therefore, we
have $o\in
Apply^\Omega(T_{\pr N \varphi}, \Delta)$ and $o\leq v_\varphi$.

Conversely, suppose $\exists o\in Apply^\Omega(T_{\pr N \varphi},
\Delta): o\leq v_\varphi$. This implies $o(\omega)=0$.
By Lemma~\ref{lm:test1}(3), this means
$
\Delta \bstep{\bar a y} \Delta'
$
and $o\in Apply^\Omega(T_{\pr N \phi}, \Delta')$.
By the induction hypothesis, we have 
$\Delta' \models \phi$, and therefore, by Definition~\ref{def:sat}, 
$\Delta \models \varphi.$

\item $\varphi = \ldia {\bar a(x)} \phi:$
This is similar to the previous case. The only difference is that
the guard $[x \not = N]$ makes sure that it is the bound output transition
that is enabled from $\Delta$, so we use Lemma~\ref{lm:test1}(4) in place of Lemma~\ref{lm:test1}(3).

\item $\varphi = \ldia {a(x)} \phi:$
Suppose $\Delta \models \varphi.$ Then for every name $w$, there exist $\Delta_1$, $\Delta_2$
and $\Delta'$ such that:
\begin{equation}
\label{eq:comp1}
\Delta \bstep{\hat \tau} \Delta_1 \sstep{a(x)} \Delta_2,
\qquad
\Delta_2[w/x] \bstep{\hat \tau} \Delta', 
\qquad
\mbox{ and }
\Delta' \models \phi[w/x].
\end{equation}
In particular, (\ref{eq:comp1}) holds for any $w \in N'$, where $N'=N\cup\{new(N)\}$.
By the induction hypothesis, $\exists o_w\in Apply^\Omega(T_{\pr {N'}
  {\phi[w/x]}}): o_w\leq v_{\pr {N'} {\phi[w/x]}}$, hence
by Lemma~\ref{lm:test1}(5), 
$$o_w\in Apply^\Omega(\omega+\bar a w.T_{\pr {N'} {\phi[w/x]}}, \Delta)$$ 
for each $w \in N'.$ Then by Lemma~\ref{lm:test1}(6), we have
$$o\in Apply^\Omega(T_{\pr {N} \varphi}, \Delta))$$
where $o=\sum_{w\in N'}p_w\cdot o_w ~\leq~ o_\varphi$.

Suppose $\exists o\in Apply^\Omega(T_{\pr {N} \varphi}, \Delta): o\leq
v_\varphi$.
Then by Lemma~\ref{lm:test1}(6), we have 
$o=\sum_{w\in N'}p_w\cdot o_w$ for some $o_w$ with
$$o_w\in Apply^\Omega(\omega+\bar{a}w.T_{\pr {N'} {\phi[w/x]}}, \Delta)$$
The latter means, by Lemma~\ref{lm:test1}(5), for each $w \in N'$, there are 
$\Delta_1$, $\Delta_2$ and $\Delta'$ such that
\begin{equation}
\label{eq:comp2}
\Delta \bstep{\hat \tau} \Delta_1 \sstep {a(x)} \Delta_2,
\qquad
\Delta_2[w/x] \bstep{\hat \tau} \Delta',
\end{equation}
and 
\begin{equation}\label{eq:com3}
o_w\in Apply^\Omega(T_{\pr{N'} {\phi[w/x]}}, \Delta').
\end{equation}
Since $\sum_{w\in N'}p_w\cdot o_w = o \leq v_\varphi = \sum_{w\in
  N'}p_w\cdot v_{\phi[w/x]}$, we have 
\begin{equation}\label{eq:com4}
o_w \leq v_{\phi[w/x]}
\end{equation} for
each $w\in N'$. Otherwise, suppose $o_w(\omega) >
v_{\phi[w/x]}(\omega)$ for some $\omega\in\Omega$. We would have
$o(\omega)=p_w\cdot o_w(\omega) > p_w\cdot v_{\phi[w/x]}(\omega) =
v_\varphi(w)$, a contradiction to $o\leq v_\varphi$.
By (\ref{eq:com3}), (\ref{eq:com4}), and the induction hypothesis, we have
\begin{equation}
\label{eq:comp3}
\Delta' \models \phi[w/x]. 
\end{equation}

To show $\Delta \models \varphi$, we need to show for every $w$, there exist
$\Delta_1$, $\Delta_2$ and $\Delta'$ satisfying (\ref{eq:comp2}) and (\ref{eq:comp3})
above. We have shown this for $w \in N'$. For the case where $w \not \in N'$, 
this is obtained from the case where $x = z$ via the renaming $[w/z]$: Recall that $z \not \in N$,
so $z \not \in fn(\Delta_2)$ and $z\not \in fn(\phi)$. Therefore, we have, from (\ref{eq:comp2}) 
and Lemma~\ref{lm:rename} (2),
$$
\Delta_2[z/x][w/z] = \Delta_2[w/x] \bstep {\hat\tau} \Delta'[w/z] 
$$ 
and from (\ref{eq:comp3}) and Lemma~\ref{lm:sat-renaming}, we have
$\Delta'[w/z] \models \phi[w/x] = \phi[z/x][w/z].$

\item $\varphi = \pand_{i \in I} \varphi_i:$
Suppose $\Delta \models \varphi$. Then $\Delta \models \phi_i$ for all $i \in I$,
and by the induction hypothesis, $o_i\in Apply^\Omega(T_{\pr N
  {\phi_i}}, \Delta): o_i\leq v_{\varphi_i}$ and by Lemma~\ref{lm:test1}(6)
$$\sum_{i\in I}p_i\cdot o_i \in Apply^\Omega(T_{\pr N \varphi},
\Delta)$$ and $\sum_{i\in I}p_i\cdot o_i \leq \sum_{i\in I}p_i\cdot v_{\varphi_i}=v_\varphi$.

Suppose $\exists o\in Apply(T_{\pr N \varphi}, \Delta): o\leq
v_\varphi$ Then by Lemma~\ref{lm:test1}(6), $o=\sum_{i\in I}p_i\cdot
o_i$ with
$$o_i\in Apply(T_{\pr N {\phi_i}}, \Delta)$$ for each $i\in I$.
As in the last case, we see from $\sum_{i\in I}p_i\cdot
o_i \leq \sum_{i\in I}p_i\cdot v_{\varphi_i} 
$ that $o_i\leq v_{\varphi_i}$ for each $i\in I$.
By induction,  we have $\Delta \models \phi_i$, therefore, by Definition~\ref{def:sat},
$\Delta \models \varphi.$

\item $\varphi = \psum_{i\in I} p_i \cdot \varphi_i:$
Suppose $\Delta \models \varphi$. Then $\Delta \bstep{\hat \tau} \sum_{i\in I} p_i \cdot \Delta_i$
and $\Delta_i \models \phi_i.$ By the induction hypothesis, 
$$
\exists o_i\in Apply^\Omega(T_i, \Delta_i): o_i\leq v_i.
$$ Hence, there are $o'_i\in Apply^\Omega(T'_i, \Delta_i)$ with
$o'_i\leq v'_i$. 
Thus by Lemma~\ref{lm:test1}(7), $o:=\sum_{i\in I}p_i\cdot o'_i \in
Apply^\Omega(T_{\pr N \varphi}, \Delta)$, and $o\leq v_\varphi$.

Conversely, suppose $\exists o\in Apply(T_{\pr N \varphi}, \Delta):
o\leq v_\varphi$.
Then by Lemma~\ref{lm:test2}, there are $q_i$ and $\Delta_i$, for all $i \in I$,
such that $\sum_{i\in I} q_i = 1$ and $\Delta \bstep{\hat \tau} \sum_{i\in I} q_i \cdot \Delta_i$
and $o=\sum_{i\in I}q_i\cdot o'_i$ for some
 $o'_i\in Apply^\Omega(T'_i, \Delta_i)$.
Now $o'_i(\omega_i)=v'_i(\omega_i)=\frac{1}{2}$ for each $i\in
I$. Using that $(T_i)_{i\in I}$ is an $\Omega$-disjoint family of
tests,
$\frac{1}{2}q_i = q_i o'_i(\omega_i) = o(\omega_i) \leq
v_\varphi(\omega_i)=p_i v'_i(\omega_i)=\frac{1}{2}p_i$. As $\sum_{i\in
I}q_i = \sum_{i\in I}p_i =1$, it must be that $q_i=p_i$ for all $i\in
I$. Exactly as in the previous case we obtain $o'_i\leq v'_i$ for all
$i\in I$. Given that $T'_i=T_i \pch{\frac{1}{2}}\omega_i$, using
Lemma~\ref{lm:test1}(6), it must be that
$o'=\frac{1}{2}o_i+\frac{1}{2}\vec{\omega_i}$ for some $o_i\in
Apply^\Omega(T_i,\Delta_i)$ with $o_i\leq v_i$.
By induction,
$\Delta_i \models \phi_i$ for all $i\in I$,  Therefore, by Definition~\ref{def:sat}, $\Delta \models \varphi.$
\end{itemize}

In case $\varphi\in\Lcal$, the formula cannot be of the form
$\Ref{X}$. Then it is easy to show that
$\sum_{\omega\in\Omega}v_\varphi(\omega)=1$ and for all $\Delta$ and
$o\in Apply^\Omega(T_\varphi,\Delta)$ we have
$\sum_{w\in\Omega}o(\omega)=1$. Therefore, $o\leq v_\varphi$ iff
$o\geq v_\varphi$ iff $o=v_\varphi$, yielding (\ref{eq:comp-ex2}).
\qed
\end{proof}
Completeness of $\pmay^\Omega$ and $\pmust^\Omega$, and hence also
$\pmay$ and $\pmust$ by Theorem~\ref{thm:modal-sim} and Theorem~\ref{thm:multi-uni}, 
follows from Lemma~\ref{lm:comp}. 

\begin{theorem}\label{thm:multi-logic}
\begin{enumerate}
\item
If $P \pmay^\Omega Q$ then $P \lleq Q.$
\item
If $P \pmust^\Omega Q$ then $P \fleq Q.$
\end{enumerate}
\end{theorem}
\begin{proof}
Suppose $P \pmay^\Omega Q$ and $\interp P \models \psi$ for some $\psi \in \Lcal.$
Let $N = fn(P, \psi)$ and let $T_{\pr N \psi}$ be a characteristic test of $\psi$ with target value $v_\psi$.
Then by Lemma~\ref{lm:comp}, we have 
$$\exists o\in Apply^\Omega(T_{\pr N \psi}, \interp P): o\geq v_\psi.$$
But since $P \pmay^\Omega Q$, this means 
$\exists o'\in Apply^\Omega(T_{\pr N \psi}, \interp Q): o\leq o'$, and thus $o'\geq v_\psi$.
So again, by Lemma~\ref{lm:comp},
we have $\interp Q \models \psi$. 

The case for must preorder is similar, using the Smyth preorder.
\qed
\end{proof}

\begin{theorem}
\begin{enumerate}
\item
If $P \pmay Q$ then $P \simpreo Q.$
\item
If $P \pmust Q$ then $P \failsimpreo Q.$
\end{enumerate}
\end{theorem}

\section{Related and future work}

There have been a number of previous works on probabilistic extensions of
the $\pi$-calculus by Palamidessi et. al. 
\cite{HerescuP00,Chatzikokolakis07,Norman09}. One distinction between
our formulation with that of Palamidessi et. al. is the fact that we
consider an interpretation of probabilistic summation as distribution over
state-based processes, whereas in those works, a process like $s {\pch p} t$
is considered as a proper process, which can evolve into the distribution
$p\cdot \pdist s + (1-p) \cdot \pdist t$ via an internal transition. We could encode
this behaviour by a simple prefixing with the $\tau$ prefix. It would be interesting
to see whether similar characterisations could be obtained for this restricted
calculus. As far as we know, there are no existing works in the literature that give 
characterisations of the may- and must-testing preorders for the probabilistic $\pi$-calculus.

We structure our completeness proofs for the simulation preorders 
along the line of the proofs of similar characterisations 
of simulation preorders for pCSP~\cite{Deng07ENTCS,Deng08LMCS}. 
The name-passing feature of the $\pi$-calculus, however, gives rise to
several complications not encountered in pCSP, and requires new
techniques to deal with. In particular, due to the possibility
of scope extrusion and close communication, the congruence properties
of (failure) simulation is proved using an adaptation of the up-to techniques~\cite{Sangiorgi98MSCS}. 

The immediate future work is to consider replication/recursion.  There
is a well-known problem with handling possible divergence; some ideas
developed in \cite{Deng09CONCUR,Boreale95IC} might be useful for
studying the semantics of $\pi_p$ as well.

\paragraph{Acknowledgment} The second author is supported by the Australian Research
Council Discovery Project DP110103173. 
Part of this work was done when the second author was visiting NICTA Kensington Lab in 2009;
he would like to thank NICTA for the support he received during his visit.


\begin{thebibliography}{10}

\bibitem{Boreale95IC}
M.~Boreale and R.~D. Nicola.
\newblock Testing equivalence for mobile processes.
\newblock {\em Inf. Comput.}, 120(2):279--303, 1995.

\bibitem{Chatzikokolakis07}
K.~Chatzikokolakis and C.~Palamidessi.
\newblock A framework for analyzing probabilistic protocols and its application
  to the partial secrets exchange.
\newblock {\em Theor. Comput. Sci.}, 389(3):512--527, 2007.

\bibitem{Nicola84}
R.~{De Nicola} and M.~Hennessy.
\newblock Testing equivalences for processes.
\newblock {\em Theor. Comput. Sci.}, 34:83--133, 1984.

\bibitem{DGMZ07}
Y.~Deng, R.~van Glabbeek, C.~Morgan, and C.~Zhang.
\newblock Scalar outcomes suffice for finitary probabilistic testing.
\newblock In {\em ESOP}, volume 4421 of {\em LNCS}, pages 363--378. Springer,
  2007.

\bibitem{Deng08LMCS}
Y.~Deng, R.~J. van Glabbeek, M.~Hennessy, and C.~Morgan.
\newblock Characterising testing preorders for finite probabilistic processes.
\newblock {\em Logical Methods in Computer Science}, 4(4), 2008.

\bibitem{Deng09CONCUR}
Y.~Deng, R.~J. van Glabbeek, M.~Hennessy, and C.~Morgan.
\newblock Testing finitary probabilistic processes.
\newblock In {\em CONCUR}, volume 5710 of {\em LNCS}, pages 274--288. Springer,
  2009.

\bibitem{Deng07ENTCS}
Y.~Deng, R.~J. van Glabbeek, M.~Hennessy, C.~Morgan, and C.~Zhang.
\newblock Remarks on testing probabilistic processes.
\newblock {\em ENTCS}, 172:359--397, 2007.

\bibitem{Ferrari95}
G.~L. Ferrari, U.~Montanari, and P.~Quaglia.
\newblock The weak late pi-calculus semantics as observation equivalence.
\newblock In {\em CONCUR}, volume 962 of {\em Lecture Notes in Computer
  Science}, pages 57--71. Springer, 1995.

\bibitem{Hansson90}
H.~Hansson and B.~Jonsson.
\newblock A calculus for communicating systems with time and probabitilies.
\newblock In {\em IEEE Real-Time Systems Symposium}, pages 278--287, 1990.

\bibitem{Hennessy82}
M.~Hennessy.
\newblock Powerdomains and nondeterministic recursive definitions.
\newblock In {\em Symposium on Programming}, volume 137 of {\em LNCS}, pages
  178--193. Springer, 1982.

\bibitem{Hennessy88}
M.~Hennessy.
\newblock {\em Algebraic Theory of Processes}.
\newblock MIT Press, 1988.

\bibitem{HerescuP00}
O.~M. Herescu and C.~Palamidessi.
\newblock Probabilistic asynchronous pi-calculus.
\newblock In {\em FoSSaCS}, volume 1784 of {\em LNCS}, pages 146--160.
  Springer, 2000.

\bibitem{Hoare85}
C.~Hoare.
\newblock {\em Communicating Sequential Processes}.
\newblock Prentice-Hall, 1985.

\bibitem{Ingolfsdottir95}
A.~Ing{\'o}lfsd{\'o}ttir.
\newblock Late and early semantics coincide for testing.
\newblock {\em Theor. Comput. Sci.}, 146(1{\&}2):341--349, 1995.

\bibitem{Milner92IC2}
R.~Milner, J.~Parrow, and D.~Walker.
\newblock A calculus of mobile processes, {II}.
\newblock {\em Inf. Comput.}, 100(1):41--77, 1992.

\bibitem{Milner93TCS}
R.~Milner, J.~Parrow, and D.~Walker.
\newblock Modal logics for mobile processes.
\newblock {\em Theor. Comput. Sci.}, 114(1):149--171, 1993.

\bibitem{Norman09}
G.~Norman, C.~Palamidessi, D.~Parker, and P.~Wu.
\newblock Model checking probabilistic and stochastic extensions of the
  pi-calculus.
\newblock {\em IEEE Trans. Software Eng.}, 35(2):209--223, 2009.

\bibitem{Sangiorgi96}
D.~Sangiorgi.
\newblock Bisimulation for higher-order process calculi.
\newblock {\em Inf. Comput.}, 131(2):141--178, 1996.

\bibitem{Sangiorgi98MSCS}
D.~Sangiorgi.
\newblock On the bisimulation proof method.
\newblock {\em Mathematical Structures in Computer Science}, 8(5):447--479,
  1998.

\bibitem{Sangiorgi01book}
D.~Sangiorgi and D.~Walker.
\newblock {\em $\pi$-Calculus: {A} Theory of Mobile Processes}.
\newblock Cambridge University Press, 2001.

\bibitem{Segala94CONCUR}
R.~Segala and N.~A. Lynch.
\newblock Probabilistic simulations for probabilistic processes.
\newblock In {\em CONCUR}, volume 836 of {\em LNCS}, pages 481--496. Springer,
  1994.

\bibitem{vanGlabbeek95}
R.~J. van Glabbeek, S.~A. Smolka, and B.~Steffen.
\newblock Reactive, generative and stratified models of probabilistic
  processes.
\newblock {\em Inf. Comput.}, 121(1):59--80, 1995.

\bibitem{vanGlabbeek96}
R.~J. van Glabbeek and W.~P. Weijland.
\newblock Branching time and abstraction in bisimulation semantics.
\newblock {\em J. ACM}, 43(3):555--600, 1996.

\bibitem{Yi92}
W.~Yi and K.~G. Larsen.
\newblock Testing probabilistic and nondeterministic processes.
\newblock In {\em PSTV}, volume C-8 of {\em IFIP Transactions}, pages 47--61.
  North-Holland, 1992.

\end{thebibliography}

\end{document}